\documentclass[12pt,a4paper]{article}
\usepackage{float,amsmath,amssymb,amsthm}
\usepackage{graphicx}
\newtheorem{thm}{Theorem}
\newtheorem{df}[thm]{Definition}
\newtheorem{lem}[thm]{Lemma}
\newtheorem{cor}[thm]{Corollary}
\newtheorem{conj}[thm]{Conjecture}
\newtheorem{prop}[thm]{Proposition}

\newtheorem{cl}{Claim}

\begin{document}

\title{Feedback game on Eulerian graphs}
\author{Naoki Matsumoto\thanks{
Research Institute for Digital Media and Content, Keio University, Kanagawa, Japan, 
E-mail: {\tt naoki.matsumo10@gmail.com}}
and
Atsuki Nagao\thanks{
Ochanomizu University, Tokyo, Japan. E-mail: {\tt a-nagao@is.ocha.ac.jp}}
}
\date{}
\maketitle

\begin{abstract}
In this paper,
we introduce a two-player impartial game on graphs, called a {\em feedback game},
which is a variant of the generalized geography.
The feedback game can be regarded as 
the undirected edge geography with an additional rule
that the first player who goes back to the starting vertex wins the game.
We consider the feedback game on Eulerian graphs 
since the game ends only by going back to the starting vertex.
We first show that 
deciding the winner of the feedback game on Eulerian graphs is PSPACE-complete in general
even if its maximum degree is at most~4.
In the latter half of the paper,
we discuss the feedback game on two subclasses of Eulerian graphs,
triangular grid graphs and toroidal grid graphs.
\end{abstract}



\section{Introduction}

All graphs considered in this paper are finite, loopless and undirected
unless otherwise mentioned.
A graph $G$ is {\em Eulerian} if each vertex of $G$ has even degree.
For other basic terminology in graph theory, we refer to~\cite{D}.

In combinatorial game theory, an impartial game has been well studied for a long time.
So far, many interesting impartial graphs have been found;
e.g., Nim~\cite{nim}, Kayles~\cite{kayles} and Poset game~\cite{poset}.
The most famous result in this area
is the Sprague-Grundy theorem~\cite{grundy,sprague} stating that
every impartial game (under the normal play convention) is equivalent to 
the Grundy value (or Nimber) which plays an important role 
to determine whether the player can win the game from a given position.
There are also many interesting games played on graphs;
Vertex Nim~\cite{GNim}, Ramsey game~\cite{ramsey}, Voronoi game~\cite{TDU2011} and so on.
For more details and other topics, 
we refer the reader to survey several 
books and articles~\cite{gamegraph1,combgamegraph,BCGww,Conway-book}.

One of most popular impartial games on graphs
is the generalized geography.
The {\em generalized geography} is a two-player game played on a directed graph $D$
whose vertices are words and 
$xy \in A(D)$ if and only if the end character of a word $x$ is the first one of $y$,
where $A(D)$ is the set of arcs of $D$.
For example, if $x$ is ``Japan" and $y$ is ``Netherlands", 
then $xy \in A(D)$ but $yx \notin A(D)$.
In this setting,
the game begins from some starting word 
and both players alternately extend a directed path using unused words.
The first player unable to extend the directed path loses.
It is known that deciding the winner of
the generalized geography is PSPACE-complete~\cite{GG1980}.
Moreover, several variants of the generalized geography have been considered,
e.g., the planar generalized geography~\cite{GG1980}, 
edge geography~\cite{Sch1978} and undirected geography~\cite{edgegeo}.
It is also known that for each above variant,
the decision problem which player wins the game is PSPACE-complete in general,
except the undirected vertex geography.

In this paper, 
we consider a new impartial game on a graph, 
called a {\em feedback game}, which is a variant of the undirected edge geography.
(We sometimes call it a {\em game} for the simplicity of the paper.)

\begin{df}[Feedback game]\label{def_game}
There are two players; Alice and Bob, starting with Alice.
For a given connected graph $G$ with a starting vertex $s$,
a token is put on $s$.
They alternately move the token on a vertex $u$ to a neighbor $v$ of $u$
and then delete an edge $uv$.
The first player who able to move the token back to $s$ 
or to an isolated vertex (after removing an edge used by the last move) wins the game.
\end{df}

In this paper, we investigate the feedback game on Eulerian graphs.
Note that if a given connected graph $G$ is Eulerian,
then the game does not end until the token goes back to the starting vertex $s$,
and further observe
that Bob always wins the feedback game on any connected bipartite Eulerian graph
(cf.~\cite{edgegeo}):
Let $G$ be a connected bipartite Eulerian graph,
and so, all vertices of $G$ are properly colored by two colors, black and white.
Without loss of generality, we may suppose that the starting vertex is colored by black.
Throughout the game on $G$, 
a token is always moved to a white (resp., black) vertex by Alice (resp., Bob).
Thus Bob necessarily wins the game.

On the other hand, for a given connected Eulerian graph $G$,
the decision problem which player wins the feedback game on $G$ is PSPACE-complete
even if the maximum degree of $G$ is at most~4 (Theorem~\ref{pspace}).
Therefore, a main study on the feedback game
is to determine the winner of the game on a connected Eulerian graph 
with more additional restrictions.

The remaining of the paper is organized as follows.
In the next section, we prove the PSPACE-completeness of the feedback game.
In Section~\ref{sec_victory},
we introduce an {\em even kernel} (resp., an {\em even kernel graph}),
first introduced in~\cite{edgegeo},
which is an useful subset (resp., subgraph)
guaranteeing the existence of a winning strategy of the second player.
In Sections~\ref{sec_trigrid} and~\ref{sec_toro_grid},
focusing on {\em triangular grid graphs} and {\em toroidal grid graphs},
we determine the winner of the feedback game on several subclasses of them.

\section{Complexity of the feedback game}\label{sec_complexity}

Because the feedback game can be seen as a variant of the undirected edge geography, 
it is a simple idea to construct a reduction from the undirected edge geography to the feedback game.

\begin{df}[Undirectred/Directed edge geography]\label{def_edge_geo}
There are two players; Alice and Bob, starting with Alice.
For a given connected undirected/directed graph $G$ with a starting vertex $s$,
a token is put on $s$.
They alternately move the token on a vertex $u$ to a neighbor/out-neighbor 
$v$ of $u$
and then delete an edge/arc $uv$.
The first player who able to move the token to an isolated vertex 
(after removing an edge/arc used by the last move) wins the game.
\end{df}

The directed edge geography is known as PSPACE-complete~\cite{Sch1978} via a reduction from TQBF, 
and the undirected edge geography is also known as 
PSPACE-complete~\cite{edgegeo} via a reduction from the directed edge geography.
Here TQBF (true quantified Boolean formula) is, 
given a quantified formula, the determination of whether there exists an assignment to the input variables such that the formula shows $1$.

The feedback game is different from these edge geographies on the winning rule.
Since a player wins when a token reaches the starting vertex, 
it is difficult to reduce from an instance of undirected edge geographies to 
that of the feedback game.
To avoid this difficulty we use the same idea about reduction from TQBF 
to the directed edge geography, and add a gadget before making the graph undirected.

\begin{thm}\label{pspace}
It is {\rm PSPACE}-complete to determine whether there exists a strategy 
that the first player wins a feedback game,
even if the given graph is Eulerian.
\end{thm}
\begin{proof}
We can see that this determination is in PSPACE, 
since we can check the winner using a DFS-like algorithm that recurs $O(|E|)$ times and uses $O(|E|)$ spaces on each recursion.

Now we reduce any instance of TQBF to an instance of determining the winner on the feedback game. 
The first step is the same as the famous reduction 
from TQBF to the directed edge geography~\cite{Sch1978} 
and we obtain a graph $H$ as an instance.
Note that $\Delta(H)=3$~\cite{GG1980}, 
where $\Delta(H)$ denotes the maximum degree of $H$,
and that the obtained graph $H$ has only one vertex $s$ 
with in-degree~0 and out-degree~2. 

By the definition of the feedback game,
the winner can also win in the view of 
the ``directed version" of the feedback game on $H$.
(Note that any player wins without going back to $s$; 
the player wins only when the other cannot move anymore.)
From now on, as shown in Figure~\ref{pseudoarc},
we use pseudo-arcs to make a reduction to 
the ``undirected version" of the feedback game~\cite{edgegeo}.

\begin{figure}[htb]
\centering
\unitlength 0.1in
\begin{picture}( 37.8500,  9.0000)(  4.6500,-16.5000)
%
\special{pn 8}%
\special{sh 1.000}%
\special{ar 600 1200 50 50  0.0000000 6.2831853}%
%
\special{pn 8}%
\special{sh 1.000}%
\special{ar 1406 1206 50 50  0.0000000 6.2831853}%
%
\special{pn 8}%
\special{pa 610 1200}%
\special{pa 1340 1200}%
\special{fp}%
\special{sh 1}%
\special{pa 1340 1200}%
\special{pa 1274 1180}%
\special{pa 1288 1200}%
\special{pa 1274 1220}%
\special{pa 1340 1200}%
\special{fp}%
\put(6.0000,-10.6000){\makebox(0,0){$a$}}%
\put(14.0500,-10.6500){\makebox(0,0){$b$}}%
%
\special{pn 8}%
\special{pa 1800 1200}%
\special{pa 2200 1200}%
\special{dt 0.045}%
\special{sh 1}%
\special{pa 2200 1200}%
\special{pa 2134 1180}%
\special{pa 2148 1200}%
\special{pa 2134 1220}%
\special{pa 2200 1200}%
\special{fp}%
%
\special{pn 8}%
\special{sh 1.000}%
\special{ar 2600 1200 50 50  0.0000000 6.2831853}%
%
\special{pn 8}%
\special{sh 1.000}%
\special{ar 3800 1200 50 50  0.0000000 6.2831853}%
%
\special{pn 8}%
\special{sh 1.000}%
\special{ar 3400 800 50 50  0.0000000 6.2831853}%
%
\special{pn 8}%
\special{sh 1.000}%
\special{ar 3000 1200 50 50  0.0000000 6.2831853}%
%
\special{pn 8}%
\special{sh 1.000}%
\special{ar 4200 1200 50 50  0.0000000 6.2831853}%
%
\special{pn 8}%
\special{pa 2600 1200}%
\special{pa 3000 1200}%
\special{fp}%
%
\special{pn 8}%
\special{sh 1.000}%
\special{ar 3200 1600 50 50  0.0000000 6.2831853}%
%
\special{pn 8}%
\special{sh 1.000}%
\special{ar 3600 1600 50 50  0.0000000 6.2831853}%
%
\special{pn 8}%
\special{pa 4200 1200}%
\special{pa 3800 1200}%
\special{fp}%
%
\special{pn 8}%
\special{pa 3800 1200}%
\special{pa 3400 800}%
\special{fp}%
%
\special{pn 8}%
\special{pa 3400 800}%
\special{pa 3000 1200}%
\special{fp}%
%
\special{pn 8}%
\special{pa 3000 1200}%
\special{pa 3200 1600}%
\special{fp}%
%
\special{pn 8}%
\special{pa 3200 1600}%
\special{pa 3600 1600}%
\special{fp}%
%
\special{pn 8}%
\special{pa 3800 1200}%
\special{pa 3600 1600}%
\special{fp}%
%
\special{pn 8}%
\special{pa 3600 1600}%
\special{pa 3400 800}%
\special{fp}%
\put(26.0000,-10.6000){\makebox(0,0){$a$}}%
\put(42.1000,-10.6000){\makebox(0,0){$b$}}%
\end{picture}%
\caption{Replacing an arc with a pseudo-arc}
\label{pseudoarc}
\end{figure}
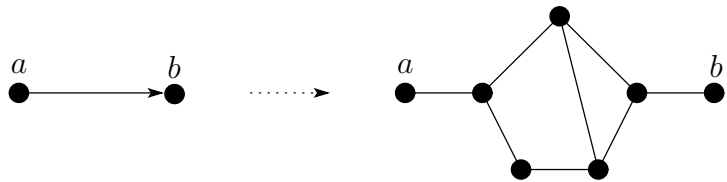

We make the undirected graph $H'$ obtained as above be Eulerian.
Let $D = \{x_1,x_2,\dots,x_{2p}\}$ for $p \geq 1$
be the set of vertices in $V(H')$ of odd degree.
First, we add a path $abc$ and two edges $as$ and $cs$, that is, $sabc$ forms a 4-cycle.
Note that the first player does not use the edge $sa$ nor $sc$ at the start on the game; 
that immediately leads to a suicide. 
Next, for each $x_i$ where $1 \leq i \leq 2p$,
we make a path $P_i=x_iy_iz_i$ of length~2 with adding two vertices $y_i$ and $z_i$.
Finally, we add edges $z_1a$, $z_2a$, $z_{2i-1}y_{2i-3}$ and $z_{2i}y_{2i-3}$, 
where $2 \leq i \leq p $.
Clearly, the resulting graph $G$ is Eulerian.
Furthermore, it is not difficult to see that 
the winner of the feedback game on $G$ is the same as that of $H'$;
note that the player who moves the token from a vertex $x_i$ 
(which is an odd in $H'$) to $y_i$ loses the game.
\end{proof}

Note that, a graph we obtain from these reductions has no vertex degree greater than $3$.
When we discuss Eulerian graphs, a graph can be added vertices and edges 
and can have vertices degree only $2$ or $4$.
Thus, we obtain the following corollary.

\begin{cor}\label{cor-pspace}
It is {\rm PSPACE}-complete to determine whether there exists a strategy 
that the first player wins a feedback game,
even if the given graph is a connected graph with maximum degree at most~$3$ or 
a connected Eulerian graph with maximum degree at most~$4$.
\end{cor}

\section{Even kernel graph}\label{sec_victory}

Remember that 
Bob wins the feedback game on every connected bipartite Eulerian graph.
Focusing on this fact, 
Fraenkel et al.~\cite{edgegeo} introduced a good concept, called an {\em even kernel}.

\begin{df}[Even kernel]\label{def_ekernel}
Let $G$ be a connected graph with a starting vertex $s$. 
An {\em even kernel} of $G$ with respect to $s$
is a non-empty subset $S \subseteq V(G)$ such that
\begin{enumerate}
\item $s \in S$,
\item no two elements of $S$ are adjacent, and
\item every vertex not in $S$ is adjacent to an even number (possibly $0$) of vertices in $S$.
\end{enumerate}
\end{df}

It is known in~\cite{Fraen1994} that 
to find an even kernel of a given graph is NP-complete
even if the graph is bipartite or its maximum degree is at most~3.
To make the even kernel be easy to handle,
we introduce a graph version concept of the even kernel, 
called an {\em even kernel graph}.
For a graph $G$ and two disjoint subsets $A,B \subseteq V(G)$,
$E_G(A,B)$ denotes the set of edges between $A$ and $B$
(i.e., one of ends of the edge in the set belongs to $A$ and the other belongs to $B$).

\begin{df}[Even kernel graph]\label{def_vic}
Let $G$ be a connected Eulerian graph with a starting vertex $s$.
An {\em even kernel graph} with respect to $s$
is a bipartite subgraph $H_s$ with the bipartition $V(H_s)=B \cup W$ 
and $E(H_s)=E_G(B,W)$, 
where $B$ is an even kernel of $G$ and 
$W$ is a superset of the set 
$N_G(B) = \{v \in V(G) \setminus B : v $ is adjacent to a vertex $u \in B\}$.
\end{df}

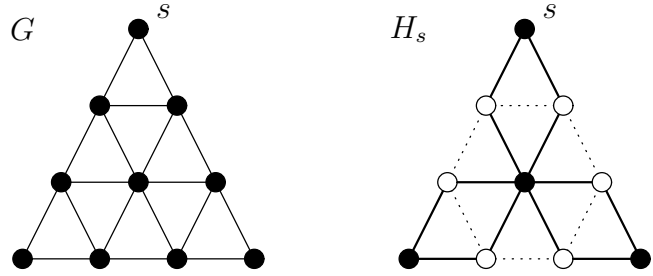
\begin{figure}[htb]
\centering
\unitlength 0.1in
\begin{picture}( 33.8500, 14.2500)( 10.6500,-20.5000)
%
\special{pn 8}%
\special{sh 1.000}%
\special{ar 1800 800 50 50  0.0000000 6.2831853}%
%
\special{pn 8}%
\special{sh 1.000}%
\special{ar 1600 1200 50 50  0.0000000 6.2831853}%
%
\special{pn 8}%
\special{sh 1.000}%
\special{ar 2000 1200 50 50  0.0000000 6.2831853}%
%
\special{pn 8}%
\special{sh 1.000}%
\special{ar 1400 1600 50 50  0.0000000 6.2831853}%
%
\special{pn 8}%
\special{sh 1.000}%
\special{ar 1800 1600 50 50  0.0000000 6.2831853}%
%
\special{pn 8}%
\special{sh 1.000}%
\special{ar 2200 1600 50 50  0.0000000 6.2831853}%
%
\special{pn 8}%
\special{sh 1.000}%
\special{ar 1200 2000 50 50  0.0000000 6.2831853}%
%
\special{pn 8}%
\special{sh 1.000}%
\special{ar 1600 2000 50 50  0.0000000 6.2831853}%
%
\special{pn 8}%
\special{sh 1.000}%
\special{ar 2000 2000 50 50  0.0000000 6.2831853}%
%
\special{pn 8}%
\special{sh 1.000}%
\special{ar 2400 2000 50 50  0.0000000 6.2831853}%
%
\special{pn 8}%
\special{pa 2400 2000}%
\special{pa 1200 2000}%
\special{fp}%
%
\special{pn 8}%
\special{pa 1200 2000}%
\special{pa 1800 800}%
\special{fp}%
%
\special{pn 8}%
\special{pa 1800 800}%
\special{pa 2400 2000}%
\special{fp}%
%
\special{pn 8}%
\special{pa 2200 1600}%
\special{pa 1400 1600}%
\special{fp}%
%
\special{pn 8}%
\special{pa 1600 1200}%
\special{pa 2000 1200}%
\special{fp}%
%
\special{pn 8}%
\special{pa 2000 1200}%
\special{pa 1600 2000}%
\special{fp}%
%
\special{pn 8}%
\special{pa 1400 1600}%
\special{pa 1600 2000}%
\special{fp}%
%
\special{pn 8}%
\special{pa 1600 1200}%
\special{pa 2000 2000}%
\special{fp}%
%
\special{pn 8}%
\special{pa 2000 2000}%
\special{pa 2200 1600}%
\special{fp}%
%
\special{pn 8}%
\special{sh 1.000}%
\special{ar 3800 800 50 50  0.0000000 6.2831853}%
%
\special{pn 8}%
\special{sh 1.000}%
\special{ar 3800 1600 50 50  0.0000000 6.2831853}%
%
\special{pn 8}%
\special{sh 1.000}%
\special{ar 3200 2000 50 50  0.0000000 6.2831853}%
%
\special{pn 8}%
\special{sh 1.000}%
\special{ar 4400 2000 50 50  0.0000000 6.2831853}%
\put(19.3000,-7.1000){\makebox(0,0){$s$}}%
\put(39.3000,-7.1000){\makebox(0,0){$s$}}%
\put(12.0000,-8.0000){\makebox(0,0){$G$}}%
\put(32.0000,-8.0000){\makebox(0,0){$H_s$}}%
%
\special{pn 13}%
\special{pa 3800 800}%
\special{pa 3600 1200}%
\special{fp}%
%
\special{pn 13}%
\special{pa 3600 1200}%
\special{pa 3800 1600}%
\special{fp}%
%
\special{pn 13}%
\special{pa 3800 1600}%
\special{pa 4000 1200}%
\special{fp}%
%
\special{pn 13}%
\special{pa 4000 1200}%
\special{pa 3800 800}%
\special{fp}%
%
\special{pn 13}%
\special{pa 3800 1600}%
\special{pa 3400 1600}%
\special{fp}%
%
\special{pn 13}%
\special{pa 3400 1600}%
\special{pa 3200 2000}%
\special{fp}%
%
\special{pn 13}%
\special{pa 3200 2000}%
\special{pa 3600 2000}%
\special{fp}%
%
\special{pn 13}%
\special{pa 3600 2000}%
\special{pa 3800 1600}%
\special{fp}%
%
\special{pn 13}%
\special{pa 3800 1600}%
\special{pa 4200 1600}%
\special{fp}%
%
\special{pn 13}%
\special{pa 4200 1600}%
\special{pa 4400 2000}%
\special{fp}%
%
\special{pn 13}%
\special{pa 4400 2000}%
\special{pa 4000 2000}%
\special{fp}%
%
\special{pn 13}%
\special{pa 4000 2000}%
\special{pa 3800 1600}%
\special{fp}%
%
\special{pn 8}%
\special{pa 4000 1200}%
\special{pa 3600 1200}%
\special{dt 0.045}%
%
\special{pn 8}%
\special{pa 3600 1200}%
\special{pa 3400 1600}%
\special{dt 0.045}%
%
\special{pn 8}%
\special{pa 3400 1600}%
\special{pa 3600 2000}%
\special{dt 0.045}%
%
\special{pn 8}%
\special{pa 3600 2000}%
\special{pa 4000 2000}%
\special{dt 0.045}%
%
\special{pn 8}%
\special{pa 4000 2000}%
\special{pa 4200 1600}%
\special{dt 0.045}%
%
\special{pn 8}%
\special{pa 4200 1600}%
\special{pa 4000 1200}%
\special{dt 0.045}%
%
\special{pn 8}%
\special{sh 0}%
\special{ar 4000 1200 50 50  0.0000000 6.2831853}%
%
\special{pn 8}%
\special{sh 0}%
\special{ar 3600 1200 50 50  0.0000000 6.2831853}%
%
\special{pn 8}%
\special{sh 0}%
\special{ar 3400 1600 50 50  0.0000000 6.2831853}%
%
\special{pn 8}%
\special{sh 0}%
\special{ar 3600 2000 50 50  0.0000000 6.2831853}%
%
\special{pn 8}%
\special{sh 0}%
\special{ar 4000 2000 50 50  0.0000000 6.2831853}%
%
\special{pn 8}%
\special{sh 0}%
\special{ar 4200 1600 50 50  0.0000000 6.2831853}%
\end{picture}%
\caption{An even kernel graph $H_s$ of a connected Eulerian graph $G$}
\label{exam_vic}
\end{figure}

For example, see Figure~\ref{exam_vic}.
The right of the figure, the graph $H_s$, is an even kernel graph of the graph $G$ 
with a starting vertex $s$.
The bold lines are edges of $H_s$ and dotted lines are ones in $E(G)\setminus E(H_s)$,
and black vertices in $B$ (where $s \in B$) and white ones in $W$.
Observe that for every vertex $v\in B$, 
all edges incident to $v$ in $G$ belong to $E(H_s)$.

\bigskip
\noindent
{\bf Remark.}
If $G$ has an even kernel, 
then $G$ always has an even kernel graph.
In Figure~\ref{exam_vic}, $H_s$ is a spanning subgraph of $G$, 
but an even kernel graph is not necessarily spanning in general.
Furthermore, 
the existence of even kernel graphs depends on the position of a starting vertex $s$.
In fact, it is easy to see that
the graph $G$ shown in Figure~\ref{exam_vic} 
has no even kernel graph if its starting vertex is of degree~4.
\bigskip

By the definition,
we see the existence of an even kernel (graph) of a connected Eulerian graph $G$
guaranteeing that Bob wins the game on $G$.

\begin{lem}[\cite{edgegeo}]\label{lem_vic}
Let $G$ be a connected graph with a starting vertex $s$. 
If $G$ has an even kernel with respect to $s$, then Bob can win the game on $G$.
\end{lem}

We conclude this section 
with showing that the converse of Lemma~\ref{lem_vic} is not true 
even if $G$ is Eulerian, that is, 
a connected Eulerian graph $G$ does not necessarily have an even kernel graph 
even if Bob can win the game on $G$.

\begin{prop}\label{prop_novic}
There exist infinitely many connected Eulerian graphs 
without an even kernel graph 
on which Bob wins the game (with respect to a prescribed starting vertex).
\end{prop}
\begin{proof}
We first give a construction of desired connected Eulerian graphs.
Prepare two even cycles $C_{2k} = u_0u_1u_2 \dots u_{2k-1}$
and $C_{4k} = v_0v_1v_2 \dots v_{4k-1}$ for some $k \geq 2$.
Add edges $u_iv_{2i}$ and $u_iv_{2i+1}$ for any $i \in \{0,1,\dots,2k-1\}$.
Finally, we add a starting vertex $s$ 
so that $s$ and $v_{j}$ are adjacent for any $j \in \{0,1,\dots,4k-1\}$.
The resulting graph is denoted by $G_k$; for example, see Figure~\ref{exam_novic}.

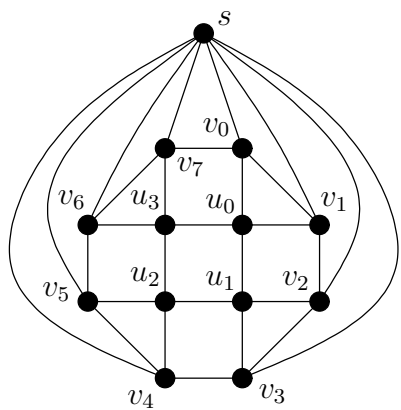
\begin{figure}[htb]
\centering
\unitlength 0.1in
\begin{picture}( 20.1000, 20.1500)( 16.0000,-24.5500)
%
\special{pn 8}%
\special{sh 1.000}%
\special{ar 2606 606 50 50  0.0000000 6.2831853}%
%
\special{pn 8}%
\special{sh 1.000}%
\special{ar 2406 1206 50 50  0.0000000 6.2831853}%
%
\special{pn 8}%
\special{sh 1.000}%
\special{ar 2806 1206 50 50  0.0000000 6.2831853}%
%
\special{pn 8}%
\special{sh 1.000}%
\special{ar 3206 1606 50 50  0.0000000 6.2831853}%
%
\special{pn 8}%
\special{sh 1.000}%
\special{ar 3206 2006 50 50  0.0000000 6.2831853}%
%
\special{pn 8}%
\special{sh 1.000}%
\special{ar 2806 2406 50 50  0.0000000 6.2831853}%
%
\special{pn 8}%
\special{sh 1.000}%
\special{ar 2406 2406 50 50  0.0000000 6.2831853}%
%
\special{pn 8}%
\special{sh 1.000}%
\special{ar 2006 2006 50 50  0.0000000 6.2831853}%
%
\special{pn 8}%
\special{sh 1.000}%
\special{ar 2006 1606 50 50  0.0000000 6.2831853}%
%
\special{pn 8}%
\special{sh 1.000}%
\special{ar 2406 1606 50 50  0.0000000 6.2831853}%
%
\special{pn 8}%
\special{sh 1.000}%
\special{ar 2406 2006 50 50  0.0000000 6.2831853}%
%
\special{pn 8}%
\special{sh 1.000}%
\special{ar 2806 2006 50 50  0.0000000 6.2831853}%
%
\special{pn 8}%
\special{sh 1.000}%
\special{ar 2806 1606 50 50  0.0000000 6.2831853}%
%
\special{pn 8}%
\special{pa 2606 606}%
\special{pa 2806 1206}%
\special{fp}%
%
\special{pn 8}%
\special{pa 2806 1206}%
\special{pa 3206 1606}%
\special{fp}%
%
\special{pn 8}%
\special{pa 3206 1606}%
\special{pa 3206 2006}%
\special{fp}%
%
\special{pn 8}%
\special{pa 3206 2006}%
\special{pa 2806 2406}%
\special{fp}%
%
\special{pn 8}%
\special{pa 2806 2406}%
\special{pa 2406 2406}%
\special{fp}%
%
\special{pn 8}%
\special{pa 2406 2406}%
\special{pa 2006 2006}%
\special{fp}%
%
\special{pn 8}%
\special{pa 2006 2006}%
\special{pa 2006 1606}%
\special{fp}%
%
\special{pn 8}%
\special{pa 2006 1606}%
\special{pa 2406 1206}%
\special{fp}%
%
\special{pn 8}%
\special{pa 2406 1206}%
\special{pa 2806 1206}%
\special{fp}%
%
\special{pn 8}%
\special{pa 2806 1206}%
\special{pa 2806 1606}%
\special{fp}%
%
\special{pn 8}%
\special{pa 2806 1606}%
\special{pa 2806 2006}%
\special{fp}%
%
\special{pn 8}%
\special{pa 2806 2006}%
\special{pa 2406 2006}%
\special{fp}%
%
\special{pn 8}%
\special{pa 2406 2006}%
\special{pa 2406 1606}%
\special{fp}%
%
\special{pn 8}%
\special{pa 2406 1606}%
\special{pa 2806 1606}%
\special{fp}%
%
\special{pn 8}%
\special{pa 2806 1606}%
\special{pa 3206 1606}%
\special{fp}%
%
\special{pn 8}%
\special{pa 3206 2006}%
\special{pa 2806 2006}%
\special{fp}%
%
\special{pn 8}%
\special{pa 2806 2006}%
\special{pa 2806 2406}%
\special{fp}%
%
\special{pn 8}%
\special{pa 2406 2406}%
\special{pa 2406 2006}%
\special{fp}%
%
\special{pn 8}%
\special{pa 2406 2006}%
\special{pa 2006 2006}%
\special{fp}%
%
\special{pn 8}%
\special{pa 2006 1606}%
\special{pa 2406 1606}%
\special{fp}%
%
\special{pn 8}%
\special{pa 2406 1606}%
\special{pa 2406 1206}%
\special{fp}%
%
\special{pn 8}%
\special{pa 2406 1206}%
\special{pa 2606 606}%
\special{fp}%
%
\special{pn 8}%
\special{pa 2006 1606}%
\special{pa 2020 1576}%
\special{pa 2032 1548}%
\special{pa 2046 1518}%
\special{pa 2060 1490}%
\special{pa 2074 1460}%
\special{pa 2088 1432}%
\special{pa 2102 1404}%
\special{pa 2116 1374}%
\special{pa 2130 1346}%
\special{pa 2146 1318}%
\special{pa 2160 1290}%
\special{pa 2176 1260}%
\special{pa 2190 1232}%
\special{pa 2206 1204}%
\special{pa 2222 1178}%
\special{pa 2238 1150}%
\special{pa 2254 1122}%
\special{pa 2270 1094}%
\special{pa 2288 1068}%
\special{pa 2304 1040}%
\special{pa 2322 1014}%
\special{pa 2338 986}%
\special{pa 2356 960}%
\special{pa 2374 934}%
\special{pa 2392 906}%
\special{pa 2410 880}%
\special{pa 2428 854}%
\special{pa 2446 828}%
\special{pa 2466 802}%
\special{pa 2484 776}%
\special{pa 2502 750}%
\special{pa 2520 722}%
\special{pa 2540 696}%
\special{pa 2558 670}%
\special{pa 2578 644}%
\special{pa 2596 618}%
\special{pa 2606 606}%
\special{sp}%
%
\special{pn 8}%
\special{pa 3206 1606}%
\special{pa 3192 1576}%
\special{pa 3178 1548}%
\special{pa 3164 1518}%
\special{pa 3150 1490}%
\special{pa 3136 1460}%
\special{pa 3122 1432}%
\special{pa 3108 1404}%
\special{pa 3094 1374}%
\special{pa 3080 1346}%
\special{pa 3066 1318}%
\special{pa 3050 1290}%
\special{pa 3036 1260}%
\special{pa 3020 1232}%
\special{pa 3006 1204}%
\special{pa 2990 1178}%
\special{pa 2974 1150}%
\special{pa 2956 1122}%
\special{pa 2940 1094}%
\special{pa 2924 1068}%
\special{pa 2906 1040}%
\special{pa 2890 1014}%
\special{pa 2872 986}%
\special{pa 2854 960}%
\special{pa 2836 934}%
\special{pa 2818 906}%
\special{pa 2800 880}%
\special{pa 2782 854}%
\special{pa 2764 828}%
\special{pa 2746 802}%
\special{pa 2728 776}%
\special{pa 2708 750}%
\special{pa 2690 722}%
\special{pa 2672 696}%
\special{pa 2652 670}%
\special{pa 2634 644}%
\special{pa 2616 618}%
\special{pa 2606 606}%
\special{sp}%
%
\special{pn 8}%
\special{pa 2606 606}%
\special{pa 2572 630}%
\special{pa 2540 652}%
\special{pa 2508 676}%
\special{pa 2476 700}%
\special{pa 2442 724}%
\special{pa 2412 748}%
\special{pa 2380 772}%
\special{pa 2348 796}%
\special{pa 2316 818}%
\special{pa 2286 842}%
\special{pa 2256 866}%
\special{pa 2226 890}%
\special{pa 2198 914}%
\special{pa 2170 938}%
\special{pa 2142 962}%
\special{pa 2114 988}%
\special{pa 2088 1012}%
\special{pa 2064 1036}%
\special{pa 2038 1060}%
\special{pa 2014 1086}%
\special{pa 1992 1110}%
\special{pa 1970 1134}%
\special{pa 1950 1160}%
\special{pa 1930 1184}%
\special{pa 1912 1210}%
\special{pa 1894 1234}%
\special{pa 1878 1260}%
\special{pa 1864 1286}%
\special{pa 1850 1312}%
\special{pa 1838 1338}%
\special{pa 1828 1364}%
\special{pa 1820 1390}%
\special{pa 1812 1416}%
\special{pa 1806 1442}%
\special{pa 1802 1468}%
\special{pa 1798 1496}%
\special{pa 1798 1522}%
\special{pa 1798 1550}%
\special{pa 1800 1578}%
\special{pa 1806 1604}%
\special{pa 1812 1632}%
\special{pa 1820 1660}%
\special{pa 1830 1688}%
\special{pa 1840 1718}%
\special{pa 1854 1746}%
\special{pa 1868 1774}%
\special{pa 1882 1802}%
\special{pa 1898 1832}%
\special{pa 1914 1860}%
\special{pa 1932 1890}%
\special{pa 1950 1918}%
\special{pa 1968 1948}%
\special{pa 1986 1976}%
\special{pa 2006 2006}%
\special{sp}%
%
\special{pn 8}%
\special{pa 2606 606}%
\special{pa 2638 630}%
\special{pa 2670 652}%
\special{pa 2704 676}%
\special{pa 2736 700}%
\special{pa 2768 724}%
\special{pa 2800 748}%
\special{pa 2832 772}%
\special{pa 2864 796}%
\special{pa 2894 818}%
\special{pa 2924 842}%
\special{pa 2954 866}%
\special{pa 2984 890}%
\special{pa 3014 914}%
\special{pa 3042 938}%
\special{pa 3070 962}%
\special{pa 3096 988}%
\special{pa 3122 1012}%
\special{pa 3148 1036}%
\special{pa 3172 1060}%
\special{pa 3196 1086}%
\special{pa 3218 1110}%
\special{pa 3240 1134}%
\special{pa 3262 1160}%
\special{pa 3280 1184}%
\special{pa 3300 1210}%
\special{pa 3316 1234}%
\special{pa 3332 1260}%
\special{pa 3346 1286}%
\special{pa 3360 1312}%
\special{pa 3372 1338}%
\special{pa 3382 1364}%
\special{pa 3392 1390}%
\special{pa 3400 1416}%
\special{pa 3406 1442}%
\special{pa 3410 1468}%
\special{pa 3412 1496}%
\special{pa 3414 1522}%
\special{pa 3412 1550}%
\special{pa 3410 1578}%
\special{pa 3406 1604}%
\special{pa 3400 1632}%
\special{pa 3392 1660}%
\special{pa 3382 1688}%
\special{pa 3370 1718}%
\special{pa 3358 1746}%
\special{pa 3344 1774}%
\special{pa 3328 1802}%
\special{pa 3314 1832}%
\special{pa 3296 1860}%
\special{pa 3278 1890}%
\special{pa 3262 1918}%
\special{pa 3242 1948}%
\special{pa 3224 1976}%
\special{pa 3206 2006}%
\special{sp}%
%
\special{pn 8}%
\special{pa 2606 606}%
\special{pa 2634 620}%
\special{pa 2662 636}%
\special{pa 2690 650}%
\special{pa 2718 666}%
\special{pa 2746 680}%
\special{pa 2774 696}%
\special{pa 2802 712}%
\special{pa 2830 728}%
\special{pa 2858 744}%
\special{pa 2886 760}%
\special{pa 2914 776}%
\special{pa 2940 792}%
\special{pa 2968 810}%
\special{pa 2994 826}%
\special{pa 3020 844}%
\special{pa 3048 862}%
\special{pa 3074 882}%
\special{pa 3100 900}%
\special{pa 3124 920}%
\special{pa 3150 940}%
\special{pa 3174 960}%
\special{pa 3198 980}%
\special{pa 3224 1002}%
\special{pa 3246 1024}%
\special{pa 3270 1048}%
\special{pa 3292 1070}%
\special{pa 3316 1094}%
\special{pa 3338 1120}%
\special{pa 3358 1146}%
\special{pa 3380 1172}%
\special{pa 3400 1198}%
\special{pa 3420 1226}%
\special{pa 3440 1256}%
\special{pa 3458 1284}%
\special{pa 3476 1314}%
\special{pa 3494 1344}%
\special{pa 3510 1376}%
\special{pa 3526 1406}%
\special{pa 3540 1438}%
\special{pa 3554 1470}%
\special{pa 3566 1502}%
\special{pa 3576 1534}%
\special{pa 3586 1564}%
\special{pa 3594 1596}%
\special{pa 3600 1628}%
\special{pa 3606 1658}%
\special{pa 3610 1688}%
\special{pa 3612 1718}%
\special{pa 3612 1748}%
\special{pa 3610 1778}%
\special{pa 3606 1806}%
\special{pa 3600 1832}%
\special{pa 3592 1858}%
\special{pa 3582 1884}%
\special{pa 3572 1910}%
\special{pa 3558 1934}%
\special{pa 3544 1958}%
\special{pa 3526 1980}%
\special{pa 3510 2002}%
\special{pa 3490 2024}%
\special{pa 3468 2046}%
\special{pa 3446 2066}%
\special{pa 3424 2086}%
\special{pa 3398 2106}%
\special{pa 3372 2124}%
\special{pa 3346 2144}%
\special{pa 3318 2162}%
\special{pa 3288 2180}%
\special{pa 3258 2196}%
\special{pa 3228 2214}%
\special{pa 3196 2230}%
\special{pa 3164 2248}%
\special{pa 3130 2264}%
\special{pa 3098 2280}%
\special{pa 3064 2296}%
\special{pa 3028 2310}%
\special{pa 2994 2326}%
\special{pa 2958 2342}%
\special{pa 2922 2356}%
\special{pa 2886 2372}%
\special{pa 2852 2386}%
\special{pa 2816 2402}%
\special{pa 2806 2406}%
\special{sp}%
%
\special{pn 8}%
\special{pa 2606 606}%
\special{pa 2578 620}%
\special{pa 2548 636}%
\special{pa 2520 650}%
\special{pa 2492 666}%
\special{pa 2464 680}%
\special{pa 2436 696}%
\special{pa 2408 712}%
\special{pa 2380 728}%
\special{pa 2352 744}%
\special{pa 2326 760}%
\special{pa 2298 776}%
\special{pa 2270 792}%
\special{pa 2244 810}%
\special{pa 2216 826}%
\special{pa 2190 844}%
\special{pa 2164 862}%
\special{pa 2138 882}%
\special{pa 2112 900}%
\special{pa 2086 920}%
\special{pa 2062 940}%
\special{pa 2036 960}%
\special{pa 2012 980}%
\special{pa 1988 1002}%
\special{pa 1964 1024}%
\special{pa 1940 1048}%
\special{pa 1918 1070}%
\special{pa 1896 1094}%
\special{pa 1874 1120}%
\special{pa 1852 1146}%
\special{pa 1832 1172}%
\special{pa 1810 1198}%
\special{pa 1790 1226}%
\special{pa 1772 1256}%
\special{pa 1752 1284}%
\special{pa 1734 1314}%
\special{pa 1718 1344}%
\special{pa 1700 1376}%
\special{pa 1686 1406}%
\special{pa 1670 1438}%
\special{pa 1658 1470}%
\special{pa 1646 1502}%
\special{pa 1634 1534}%
\special{pa 1626 1564}%
\special{pa 1616 1596}%
\special{pa 1610 1628}%
\special{pa 1606 1658}%
\special{pa 1602 1688}%
\special{pa 1600 1718}%
\special{pa 1600 1748}%
\special{pa 1602 1778}%
\special{pa 1606 1806}%
\special{pa 1612 1832}%
\special{pa 1618 1858}%
\special{pa 1628 1884}%
\special{pa 1640 1910}%
\special{pa 1652 1934}%
\special{pa 1668 1958}%
\special{pa 1684 1980}%
\special{pa 1702 2002}%
\special{pa 1722 2024}%
\special{pa 1742 2046}%
\special{pa 1764 2066}%
\special{pa 1788 2086}%
\special{pa 1812 2106}%
\special{pa 1838 2124}%
\special{pa 1866 2144}%
\special{pa 1894 2162}%
\special{pa 1922 2180}%
\special{pa 1952 2196}%
\special{pa 1984 2214}%
\special{pa 2016 2230}%
\special{pa 2048 2248}%
\special{pa 2080 2264}%
\special{pa 2114 2280}%
\special{pa 2148 2296}%
\special{pa 2182 2310}%
\special{pa 2218 2326}%
\special{pa 2252 2342}%
\special{pa 2288 2356}%
\special{pa 2324 2372}%
\special{pa 2360 2386}%
\special{pa 2396 2402}%
\special{pa 2406 2406}%
\special{sp}%
\put(27.1500,-5.2500){\makebox(0,0){$s$}}%
\put(26.9500,-14.9500){\makebox(0,0){$u_0$}}%
\put(26.9500,-18.9500){\makebox(0,0){$u_1$}}%
\put(23.0500,-18.7500){\makebox(0,0){$u_2$}}%
\put(23.0500,-14.7500){\makebox(0,0){$u_3$}}%
\put(26.7500,-10.9500){\makebox(0,0){$v_0$}}%
\put(32.8500,-14.8500){\makebox(0,0){$v_1$}}%
\put(30.8500,-18.9500){\makebox(0,0){$v_2$}}%
\put(29.6500,-24.8500){\makebox(0,0){$v_3$}}%
\put(22.8500,-25.1500){\makebox(0,0){$v_4$}}%
\put(18.4500,-19.5500){\makebox(0,0){$v_5$}}%
\put(19.2500,-14.7500){\makebox(0,0){$v_6$}}%
\put(25.4000,-13.0000){\makebox(0,0){$v_7$}}%
\end{picture}%
\caption{The graph $G_2$}
\label{exam_novic}
\end{figure}

We next show that Bob can win the game on $G_k$.
Without loss of generality, 
we may suppose that Alice first moves the token to $v_0$
and that Bob moves it from $v_0$ to $u_0$.
If Alice moves the token to $v_1$, then Bob wins the game.
Thus we may assume that Alice moves it to $u_1$,
and then Bob moves it to $u_2$.
After that,
Alice (resp., Bob) moves the token from $u_{2i}$ to $u_{2i+1}$
(resp., from $u_{2i+1}$ to $u_{2i+2}$),
where subscripts are modulo $2k$.
Therefore, Bob finally moves the token to $u_0$, that is, Alice has to move it to $v_1$.
Thus, Bob wins the game on $G_k$.

Finally, we claim that $G_k$ has no even kernel graph with respect to $s$.
Suppose to the contrary that $G_k$ has an even kernel graph $H_s$ with 
bipartite sets $B$ and $W$ where $s \in B$.
By the definition of an even kernel graph, 
$sv_i \in E(H_s)$ for all $i \in \{0,1,\dots,4k-1\}$, that is, $v_i \in W$.
Since $H_s$ is bipartite, $v_iv_{i+1} \notin E(H_s)$ where subscripts are modulo $4k$.
Thus all edges between two cycles $C_{4k}$ and $C_{2k}$ belong to $E(H_s)$,
and hence, $u_j \in B$ for any $j \in \{0,1,\dots,2k-1\}$.
However, $u_0$ and $u_1$ must be adjacent in $H_s$, 
which contradicts the bipartiteness of $H_s$.
\end{proof}

\section{Triangular grid graphs}\label{sec_trigrid}

At first, we give a recursive definition of triangular grid graphs.

\begin{df}[Triangular grid graph]
A {\em triangular grid graph} $T_n$ with $n\geq 0$ is recursively constructed as follows.
\begin{itemize}
\item $T_0 \, (= P^0)$ consists of an isolated vertex $v^0_0$ and no edge.
\item $T_n$ with $n\geq 1$ is obtained from $T_{n-1}$
by adding a path $P^{n} = v^n_0v^n_1 \dots v^n_{n}$
and edges $v^n_0v^{n-1}_0$, $v^n_nv^{n-1}_{n-1}$,
$v^n_iv^{n-1}_{i-1}$ and $v^n_iv^{n-1}_{i}$ for any $i \in \{1,\dots,n-1\}$.
\end{itemize}
\end{df}

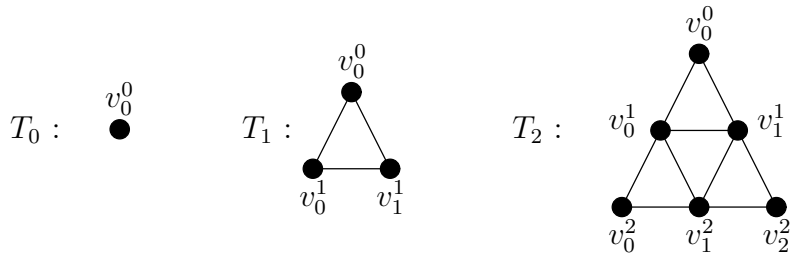
\begin{figure}[htb]
\centering
\unitlength 0.1in
\begin{picture}( 42.1000, 11.1000)( 10.5000,-22.7000)
%
\special{pn 8}%
\special{sh 1.000}%
\special{ar 1810 1800 50 50  0.0000000 6.2831853}%
%
\special{pn 8}%
\special{sh 1.000}%
\special{ar 3010 1606 50 50  0.0000000 6.2831853}%
%
\special{pn 8}%
\special{sh 1.000}%
\special{ar 2810 2006 50 50  0.0000000 6.2831853}%
%
\special{pn 8}%
\special{sh 1.000}%
\special{ar 3210 2006 50 50  0.0000000 6.2831853}%
%
\special{pn 8}%
\special{sh 1.000}%
\special{ar 4410 2206 50 50  0.0000000 6.2831853}%
%
\special{pn 8}%
\special{sh 1.000}%
\special{ar 4810 2206 50 50  0.0000000 6.2831853}%
%
\special{pn 8}%
\special{sh 1.000}%
\special{ar 5210 2206 50 50  0.0000000 6.2831853}%
%
\special{pn 8}%
\special{sh 1.000}%
\special{ar 5010 1806 50 50  0.0000000 6.2831853}%
%
\special{pn 8}%
\special{sh 1.000}%
\special{ar 4610 1806 50 50  0.0000000 6.2831853}%
%
\special{pn 8}%
\special{sh 1.000}%
\special{ar 4810 1406 50 50  0.0000000 6.2831853}%
%
\special{pn 8}%
\special{pa 4810 1406}%
\special{pa 4410 2206}%
\special{fp}%
%
\special{pn 8}%
\special{pa 4410 2206}%
\special{pa 5210 2206}%
\special{fp}%
%
\special{pn 8}%
\special{pa 5210 2206}%
\special{pa 4810 1406}%
\special{fp}%
%
\special{pn 8}%
\special{pa 5010 1806}%
\special{pa 4610 1806}%
\special{fp}%
%
\special{pn 8}%
\special{pa 4610 1806}%
\special{pa 4810 2206}%
\special{fp}%
%
\special{pn 8}%
\special{pa 4810 2206}%
\special{pa 5010 1806}%
\special{fp}%
%
\special{pn 8}%
\special{pa 3210 2006}%
\special{pa 2810 2006}%
\special{fp}%
%
\special{pn 8}%
\special{pa 2810 2006}%
\special{pa 3010 1606}%
\special{fp}%
%
\special{pn 8}%
\special{pa 3010 1606}%
\special{pa 3210 2006}%
\special{fp}%
\put(18.1000,-16.4000){\makebox(0,0){$v^0_0$}}%
\put(30.1500,-14.4500){\makebox(0,0){$v^0_0$}}%
\put(48.1500,-12.4500){\makebox(0,0){$v^0_0$}}%
\put(28.1500,-21.5500){\makebox(0,0){$v^1_0$}}%
\put(32.1500,-21.5500){\makebox(0,0){$v^1_1$}}%
\put(51.8500,-17.5500){\makebox(0,0){$v^1_1$}}%
\put(44.1500,-17.5500){\makebox(0,0){$v^1_0$}}%
\put(44.1500,-23.5500){\makebox(0,0){$v^2_0$}}%
\put(48.1500,-23.5500){\makebox(0,0){$v^2_1$}}%
\put(52.1500,-23.5500){\makebox(0,0){$v^2_2$}}%
\put(14.1000,-18.0000){\makebox(0,0){$T_0$ : }}%
\put(26.1000,-18.0000){\makebox(0,0){$T_1$ : }}%
\put(40.1000,-18.0000){\makebox(0,0){$T_2$ : }}%
\end{picture}%
\caption{Triangular grid graphs $T_0, T_1$ and $T_2$}
\label{tgrid}
\end{figure}

For example, see Figure~\ref{tgrid}.
It is easy to see that every triangular grid graph is connected and Eulerian
and that its maximum degree is at most~6.
Moreover, it has high symmetry as we know.
Thus the class of triangular grid graphs seems to be a reasonable subclass of 
connected Eulerian graphs
for considering the feedback game.

For triangular grid graphs, 
we have the following setting $v^0_0$ as a starting vertex 
(where note that the vertex $v^0_0$ can be regarded as $v^n_0$ and $v^n_n$ by symmetry).

\begin{thm}\label{thm_trigrid_bob_wins}
If $n \ne 2^m - 3$ with $m\geq 2$, 
then Bob wins the game on the triangular grid graph $T_n$
with a starting vertex $v^0_0$.
\end{thm}
\begin{proof}
We prove the theorem by induction on $n$.
For the base case, we can easily find that each of 
$T_2$ (the left of Figure~\ref{fourtris}), 
$T_3$ (the right of Figure~\ref{exam_vic}),
$T_4,T_6$ (Figure~\ref{fig:t4t6})
has at least one even kernel graph, 
i.e., Bob wins the game on these triangular grid graphs by Lemma~\ref{lem_vic}.

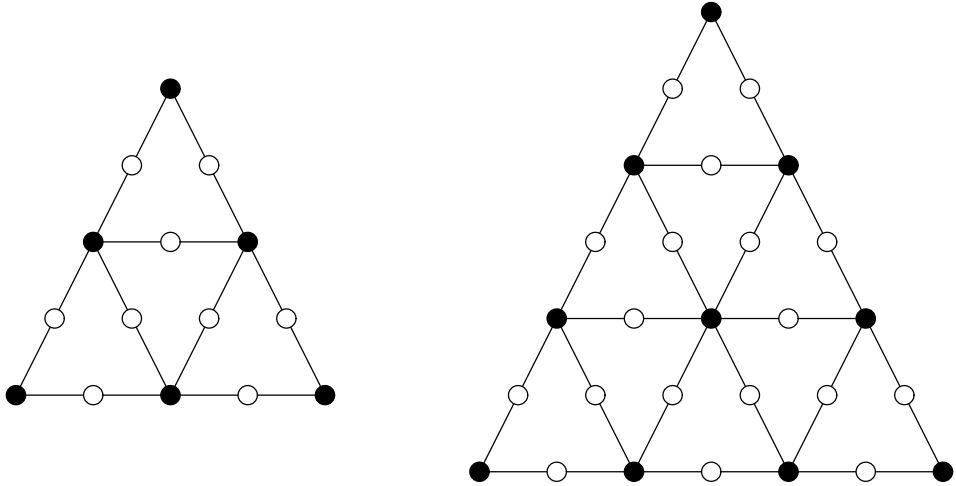
\begin{figure}[htb]
\centering
\unitlength 0.1in
\begin{picture}( 49.0000, 25.0000)( 17.5000,-32.5000)
%
\special{pn 8}%
\special{sh 1.000}%
\special{ar 2600 1200 50 50  0.0000000 6.2831853}%
%
\special{pn 8}%
\special{sh 1.000}%
\special{ar 2200 2000 50 50  0.0000000 6.2831853}%
%
\special{pn 8}%
\special{sh 1.000}%
\special{ar 3000 2000 50 50  0.0000000 6.2831853}%
%
\special{pn 8}%
\special{sh 1.000}%
\special{ar 1800 2800 50 50  0.0000000 6.2831853}%
%
\special{pn 8}%
\special{sh 1.000}%
\special{ar 2600 2800 50 50  0.0000000 6.2831853}%
%
\special{pn 8}%
\special{sh 1.000}%
\special{ar 3400 2800 50 50  0.0000000 6.2831853}%
%
\special{pn 8}%
\special{pa 2600 1200}%
\special{pa 2200 2000}%
\special{fp}%
%
\special{pn 8}%
\special{pa 2200 2000}%
\special{pa 1800 2800}%
\special{fp}%
%
\special{pn 8}%
\special{pa 1800 2800}%
\special{pa 3400 2800}%
\special{fp}%
%
\special{pn 8}%
\special{pa 3400 2800}%
\special{pa 2600 1200}%
\special{fp}%
%
\special{pn 8}%
\special{pa 2200 2000}%
\special{pa 3000 2000}%
\special{fp}%
%
\special{pn 8}%
\special{pa 3000 2000}%
\special{pa 2600 2800}%
\special{fp}%
%
\special{pn 8}%
\special{pa 2600 2800}%
\special{pa 2200 2000}%
\special{fp}%
%
\special{pn 8}%
\special{sh 0}%
\special{ar 2600 2000 50 50  0.0000000 6.2831853}%
%
\special{pn 8}%
\special{sh 0}%
\special{ar 2400 1600 50 50  0.0000000 6.2831853}%
%
\special{pn 8}%
\special{sh 0}%
\special{ar 2800 1600 50 50  0.0000000 6.2831853}%
%
\special{pn 8}%
\special{sh 0}%
\special{ar 2800 2400 50 50  0.0000000 6.2831853}%
%
\special{pn 8}%
\special{sh 0}%
\special{ar 2400 2400 50 50  0.0000000 6.2831853}%
%
\special{pn 8}%
\special{sh 0}%
\special{ar 2000 2400 50 50  0.0000000 6.2831853}%
%
\special{pn 8}%
\special{sh 0}%
\special{ar 3200 2400 50 50  0.0000000 6.2831853}%
%
\special{pn 8}%
\special{sh 0}%
\special{ar 3000 2800 50 50  0.0000000 6.2831853}%
%
\special{pn 8}%
\special{sh 0}%
\special{ar 2200 2800 50 50  0.0000000 6.2831853}%
%
\special{pn 8}%
\special{sh 1.000}%
\special{ar 5400 800 50 50  0.0000000 6.2831853}%
%
\special{pn 8}%
\special{sh 1.000}%
\special{ar 5000 1600 50 50  0.0000000 6.2831853}%
%
\special{pn 8}%
\special{sh 1.000}%
\special{ar 5800 1600 50 50  0.0000000 6.2831853}%
%
\special{pn 8}%
\special{sh 1.000}%
\special{ar 4600 2400 50 50  0.0000000 6.2831853}%
%
\special{pn 8}%
\special{sh 1.000}%
\special{ar 5400 2400 50 50  0.0000000 6.2831853}%
%
\special{pn 8}%
\special{sh 1.000}%
\special{ar 6200 2400 50 50  0.0000000 6.2831853}%
%
\special{pn 8}%
\special{pa 5400 800}%
\special{pa 5000 1600}%
\special{fp}%
%
\special{pn 8}%
\special{pa 5000 1600}%
\special{pa 4600 2400}%
\special{fp}%
%
\special{pn 8}%
\special{pa 4600 2400}%
\special{pa 6200 2400}%
\special{fp}%
%
\special{pn 8}%
\special{pa 6200 2400}%
\special{pa 5400 800}%
\special{fp}%
%
\special{pn 8}%
\special{pa 5000 1600}%
\special{pa 5800 1600}%
\special{fp}%
%
\special{pn 8}%
\special{pa 5800 1600}%
\special{pa 5400 2400}%
\special{fp}%
%
\special{pn 8}%
\special{pa 5400 2400}%
\special{pa 5000 1600}%
\special{fp}%
%
\special{pn 8}%
\special{sh 0}%
\special{ar 5400 1600 50 50  0.0000000 6.2831853}%
%
\special{pn 8}%
\special{sh 0}%
\special{ar 5200 1200 50 50  0.0000000 6.2831853}%
%
\special{pn 8}%
\special{sh 0}%
\special{ar 5600 1200 50 50  0.0000000 6.2831853}%
%
\special{pn 8}%
\special{sh 0}%
\special{ar 5600 2000 50 50  0.0000000 6.2831853}%
%
\special{pn 8}%
\special{sh 0}%
\special{ar 5200 2000 50 50  0.0000000 6.2831853}%
%
\special{pn 8}%
\special{sh 0}%
\special{ar 4800 2000 50 50  0.0000000 6.2831853}%
%
\special{pn 8}%
\special{sh 0}%
\special{ar 6000 2000 50 50  0.0000000 6.2831853}%
%
\special{pn 8}%
\special{sh 0}%
\special{ar 5800 2400 50 50  0.0000000 6.2831853}%
%
\special{pn 8}%
\special{sh 0}%
\special{ar 5000 2400 50 50  0.0000000 6.2831853}%
%
\special{pn 8}%
\special{sh 1.000}%
\special{ar 4200 3200 50 50  0.0000000 6.2831853}%
%
\special{pn 8}%
\special{sh 1.000}%
\special{ar 5000 3200 50 50  0.0000000 6.2831853}%
%
\special{pn 8}%
\special{sh 1.000}%
\special{ar 5800 3200 50 50  0.0000000 6.2831853}%
%
\special{pn 8}%
\special{sh 1.000}%
\special{ar 6600 3200 50 50  0.0000000 6.2831853}%
%
\special{pn 8}%
\special{pa 6600 3200}%
\special{pa 6200 2400}%
\special{fp}%
%
\special{pn 8}%
\special{pa 6200 2400}%
\special{pa 5800 3200}%
\special{fp}%
%
\special{pn 8}%
\special{pa 5800 3200}%
\special{pa 5400 2400}%
\special{fp}%
%
\special{pn 8}%
\special{pa 5400 2400}%
\special{pa 5000 3200}%
\special{fp}%
%
\special{pn 8}%
\special{pa 5000 3200}%
\special{pa 4600 2400}%
\special{fp}%
%
\special{pn 8}%
\special{pa 4600 2400}%
\special{pa 4200 3200}%
\special{fp}%
%
\special{pn 8}%
\special{pa 6600 3200}%
\special{pa 4200 3200}%
\special{fp}%
%
\special{pn 8}%
\special{sh 0}%
\special{ar 4400 2800 50 50  0.0000000 6.2831853}%
%
\special{pn 8}%
\special{sh 0}%
\special{ar 4600 3200 50 50  0.0000000 6.2831853}%
%
\special{pn 8}%
\special{sh 0}%
\special{ar 4800 2800 50 50  0.0000000 6.2831853}%
%
\special{pn 8}%
\special{sh 0}%
\special{ar 5200 2800 50 50  0.0000000 6.2831853}%
%
\special{pn 8}%
\special{sh 0}%
\special{ar 5600 2800 50 50  0.0000000 6.2831853}%
%
\special{pn 8}%
\special{sh 0}%
\special{ar 6000 2800 50 50  0.0000000 6.2831853}%
%
\special{pn 8}%
\special{sh 0}%
\special{ar 6400 2800 50 50  0.0000000 6.2831853}%
%
\special{pn 8}%
\special{sh 0}%
\special{ar 6200 3200 50 50  0.0000000 6.2831853}%
%
\special{pn 8}%
\special{sh 0}%
\special{ar 5400 3200 50 50  0.0000000 6.2831853}%
\end{picture}%
\caption{Even kernel graphs of $T_4$ and $T_6$}
\label{fig:t4t6}
\end{figure}

For an induction rule, 
we assume that each of $T_{2^i-2},T_{2^i-1},\dots,T_{2^{i+1}-4},T_{2^{i+1}-2}$ 
has at least one even kernel graph.
Here we construct even kernel graphs on triangular grid graphs using 
those even kernel graphs.
Using four even kernel graphs on $T_\alpha$, 
we can construct an even kernel graph on $T_{2\alpha+3}$;
for example, see Figure~\ref{fourtris}.

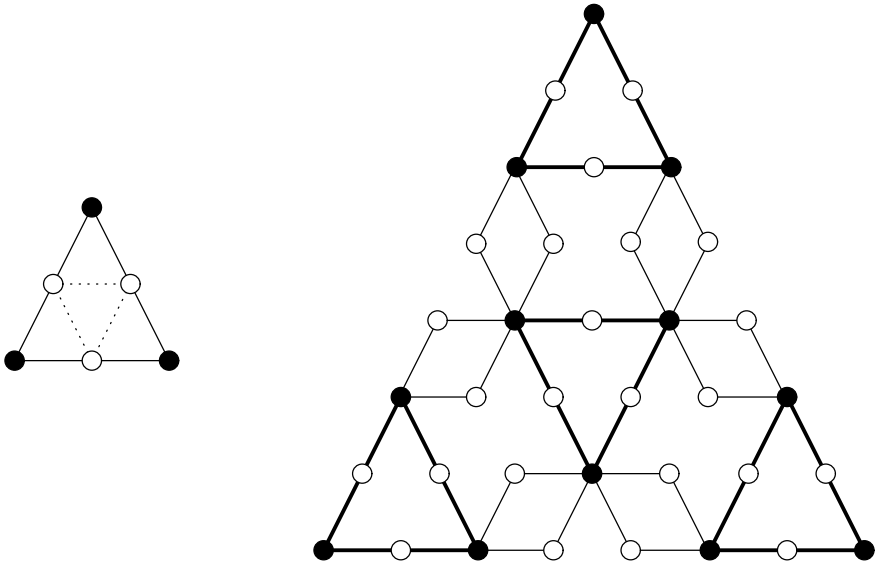
\begin{figure}[htb]
\centering
\unitlength 0.1in
\begin{picture}( 45.0000, 29.0000)(  7.6000,-30.5000)
%
\special{pn 8}%
\special{sh 1.000}%
\special{ar 810 2010 50 50  0.0000000 6.2831853}%
%
\special{pn 8}%
\special{sh 1.000}%
\special{ar 1210 1210 50 50  0.0000000 6.2831853}%
%
\special{pn 8}%
\special{sh 1.000}%
\special{ar 1610 2010 50 50  0.0000000 6.2831853}%
%
\special{pn 8}%
\special{pa 810 2010}%
\special{pa 1610 2010}%
\special{fp}%
%
\special{pn 8}%
\special{pa 1610 2010}%
\special{pa 1210 1210}%
\special{fp}%
%
\special{pn 8}%
\special{pa 810 2010}%
\special{pa 1210 1210}%
\special{fp}%
%
\special{pn 8}%
\special{pa 1210 2010}%
\special{pa 1010 1610}%
\special{dt 0.045}%
\special{pa 1410 1610}%
\special{pa 1410 1610}%
\special{dt 0.045}%
%
\special{pn 8}%
\special{pa 1410 1610}%
\special{pa 1210 2010}%
\special{dt 0.045}%
%
\special{pn 8}%
\special{pa 1010 1610}%
\special{pa 1410 1610}%
\special{dt 0.045}%
%
\special{pn 8}%
\special{sh 0}%
\special{ar 1210 2010 50 50  0.0000000 6.2831853}%
%
\special{pn 8}%
\special{sh 0}%
\special{ar 1010 1610 50 50  0.0000000 6.2831853}%
%
\special{pn 8}%
\special{sh 0}%
\special{ar 1410 1610 50 50  0.0000000 6.2831853}%
%
\special{pn 8}%
\special{sh 1.000}%
\special{ar 2410 3000 50 50  0.0000000 6.2831853}%
%
\special{pn 8}%
\special{sh 1.000}%
\special{ar 2810 2200 50 50  0.0000000 6.2831853}%
%
\special{pn 8}%
\special{sh 1.000}%
\special{ar 3210 3000 50 50  0.0000000 6.2831853}%
%
\special{pn 20}%
\special{pa 2410 3000}%
\special{pa 3210 3000}%
\special{fp}%
%
\special{pn 20}%
\special{pa 3210 3000}%
\special{pa 2810 2200}%
\special{fp}%
%
\special{pn 20}%
\special{pa 2410 3000}%
\special{pa 2810 2200}%
\special{fp}%
%
\special{pn 8}%
\special{sh 0}%
\special{ar 2810 3000 50 50  0.0000000 6.2831853}%
%
\special{pn 8}%
\special{sh 0}%
\special{ar 2610 2600 50 50  0.0000000 6.2831853}%
%
\special{pn 8}%
\special{sh 0}%
\special{ar 3010 2600 50 50  0.0000000 6.2831853}%
%
\special{pn 8}%
\special{sh 1.000}%
\special{ar 4410 3000 50 50  0.0000000 6.2831853}%
%
\special{pn 8}%
\special{sh 1.000}%
\special{ar 4810 2200 50 50  0.0000000 6.2831853}%
%
\special{pn 8}%
\special{sh 1.000}%
\special{ar 5210 3000 50 50  0.0000000 6.2831853}%
%
\special{pn 20}%
\special{pa 4410 3000}%
\special{pa 5210 3000}%
\special{fp}%
%
\special{pn 20}%
\special{pa 5210 3000}%
\special{pa 4810 2200}%
\special{fp}%
%
\special{pn 20}%
\special{pa 4410 3000}%
\special{pa 4810 2200}%
\special{fp}%
%
\special{pn 8}%
\special{sh 0}%
\special{ar 4810 3000 50 50  0.0000000 6.2831853}%
%
\special{pn 8}%
\special{sh 0}%
\special{ar 4610 2600 50 50  0.0000000 6.2831853}%
%
\special{pn 8}%
\special{sh 0}%
\special{ar 5010 2600 50 50  0.0000000 6.2831853}%
%
\special{pn 8}%
\special{sh 1.000}%
\special{ar 3410 1000 50 50  0.0000000 6.2831853}%
%
\special{pn 8}%
\special{sh 1.000}%
\special{ar 3810 200 50 50  0.0000000 6.2831853}%
%
\special{pn 8}%
\special{sh 1.000}%
\special{ar 4210 1000 50 50  0.0000000 6.2831853}%
%
\special{pn 20}%
\special{pa 3410 1000}%
\special{pa 4210 1000}%
\special{fp}%
%
\special{pn 20}%
\special{pa 4210 1000}%
\special{pa 3810 200}%
\special{fp}%
%
\special{pn 20}%
\special{pa 3410 1000}%
\special{pa 3810 200}%
\special{fp}%
%
\special{pn 8}%
\special{sh 0}%
\special{ar 3810 1000 50 50  0.0000000 6.2831853}%
%
\special{pn 8}%
\special{sh 0}%
\special{ar 3610 600 50 50  0.0000000 6.2831853}%
%
\special{pn 8}%
\special{sh 0}%
\special{ar 4010 600 50 50  0.0000000 6.2831853}%
%
\special{pn 8}%
\special{sh 1.000}%
\special{ar 4200 1800 50 50  0.0000000 6.2831853}%
%
\special{pn 8}%
\special{sh 1.000}%
\special{ar 3800 2600 50 50  0.0000000 6.2831853}%
%
\special{pn 8}%
\special{sh 1.000}%
\special{ar 3400 1800 50 50  0.0000000 6.2831853}%
%
\special{pn 20}%
\special{pa 4200 1800}%
\special{pa 3400 1800}%
\special{fp}%
%
\special{pn 20}%
\special{pa 3400 1800}%
\special{pa 3800 2600}%
\special{fp}%
%
\special{pn 20}%
\special{pa 4200 1800}%
\special{pa 3800 2600}%
\special{fp}%
%
\special{pn 8}%
\special{sh 0}%
\special{ar 3800 1800 50 50  0.0000000 6.2831853}%
%
\special{pn 8}%
\special{sh 0}%
\special{ar 4000 2200 50 50  0.0000000 6.2831853}%
%
\special{pn 8}%
\special{sh 0}%
\special{ar 3600 2200 50 50  0.0000000 6.2831853}%
%
\special{pn 8}%
\special{pa 3200 3000}%
\special{pa 3600 3000}%
\special{fp}%
%
\special{pn 8}%
\special{pa 3600 3000}%
\special{pa 3800 2600}%
\special{fp}%
%
\special{pn 8}%
\special{pa 3800 2600}%
\special{pa 4000 3000}%
\special{fp}%
%
\special{pn 8}%
\special{pa 4000 3000}%
\special{pa 4400 3000}%
\special{fp}%
%
\special{pn 8}%
\special{pa 4400 3000}%
\special{pa 4200 2600}%
\special{fp}%
%
\special{pn 8}%
\special{pa 4200 2600}%
\special{pa 3800 2600}%
\special{fp}%
%
\special{pn 8}%
\special{pa 3200 3000}%
\special{pa 3400 2600}%
\special{fp}%
%
\special{pn 8}%
\special{pa 3400 2600}%
\special{pa 3800 2600}%
\special{fp}%
%
\special{pn 8}%
\special{pa 2800 2200}%
\special{pa 3000 1800}%
\special{fp}%
%
\special{pn 8}%
\special{pa 3000 1800}%
\special{pa 3400 1800}%
\special{fp}%
%
\special{pn 8}%
\special{pa 3400 1800}%
\special{pa 3200 2200}%
\special{fp}%
%
\special{pn 8}%
\special{pa 3200 2200}%
\special{pa 2800 2200}%
\special{fp}%
%
\special{pn 8}%
\special{pa 3400 1800}%
\special{pa 3200 1400}%
\special{fp}%
%
\special{pn 8}%
\special{pa 3200 1400}%
\special{pa 3400 1000}%
\special{fp}%
%
\special{pn 8}%
\special{pa 3400 1000}%
\special{pa 3600 1400}%
\special{fp}%
%
\special{pn 8}%
\special{pa 3600 1400}%
\special{pa 3400 1800}%
\special{fp}%
%
\special{pn 8}%
\special{pa 4200 1800}%
\special{pa 4000 1400}%
\special{fp}%
%
\special{pn 8}%
\special{pa 4000 1400}%
\special{pa 4200 1000}%
\special{fp}%
%
\special{pn 8}%
\special{pa 4200 1000}%
\special{pa 4400 1400}%
\special{fp}%
%
\special{pn 8}%
\special{pa 4400 1400}%
\special{pa 4200 1800}%
\special{fp}%
%
\special{pn 8}%
\special{pa 4200 1800}%
\special{pa 4600 1800}%
\special{fp}%
%
\special{pn 8}%
\special{pa 4600 1800}%
\special{pa 4800 2200}%
\special{fp}%
%
\special{pn 8}%
\special{pa 4800 2200}%
\special{pa 4400 2200}%
\special{fp}%
%
\special{pn 8}%
\special{pa 4400 2200}%
\special{pa 4200 1800}%
\special{fp}%
%
\special{pn 8}%
\special{sh 0}%
\special{ar 3200 2200 50 50  0.0000000 6.2831853}%
%
\special{pn 8}%
\special{sh 0}%
\special{ar 3600 3000 50 50  0.0000000 6.2831853}%
%
\special{pn 8}%
\special{sh 0}%
\special{ar 3400 2600 50 50  0.0000000 6.2831853}%
%
\special{pn 8}%
\special{sh 0}%
\special{ar 4200 2600 50 50  0.0000000 6.2831853}%
%
\special{pn 8}%
\special{sh 0}%
\special{ar 4400 2200 50 50  0.0000000 6.2831853}%
%
\special{pn 8}%
\special{sh 0}%
\special{ar 4600 1800 50 50  0.0000000 6.2831853}%
%
\special{pn 8}%
\special{sh 0}%
\special{ar 4400 1390 50 50  0.0000000 6.2831853}%
%
\special{pn 8}%
\special{sh 0}%
\special{ar 4000 1390 50 50  0.0000000 6.2831853}%
%
\special{pn 8}%
\special{sh 0}%
\special{ar 3000 1800 50 50  0.0000000 6.2831853}%
%
\special{pn 8}%
\special{sh 0}%
\special{ar 3200 1400 50 50  0.0000000 6.2831853}%
%
\special{pn 8}%
\special{sh 0}%
\special{ar 3600 1400 50 50  0.0000000 6.2831853}%
%
\special{pn 8}%
\special{sh 0}%
\special{ar 4000 3000 50 50  0.0000000 6.2831853}%
\end{picture}%
\caption{An even kernel graph $H$ of $T_2$ and that of $T_7$ based on $H$}
\label{fourtris}
\end{figure}

From the assumption and this fact, 
each of $T_{2^{i+1}-1},T_{2^{i+1}+1},\dots,T_{2^{i+2}-5},T_{2^{i+2}-1}$
has at least one even kernel graph.
For triangular grid graphs $T_{2^{i+1}-2},T_{2^{i+1}},\dots,T_{2^{i+2}-4},T_{2^{i+2}-2}$, 
it is clear that
they have an even kernel graph with bipartite sets 
$B = \{v^j_k : j \equiv k \equiv 0 \pmod{2}\}$
and
$W = \{v^j_k : j \equiv 1 \pmod{2} \: {\rm or} \: k \equiv 1 \pmod{2}\}$
since their height is even (as shown in Figure~\ref{fig:t4t6});
note that in every even kernel graph constructed above,
all vertices of degree~2 are in the same partite set.
Then, all triangular grid graphs $T_{2^{i+1}-2},T_{2^{i+1}-1},\dots,T_{2^{i+2}-4},T_{2^{i+2}-2}$ 
have at least one even kernel graph.
By induction, 
all triangular grid graph $T_n$ has at least one even kernel graph when $n\neq 2^m-3$.
This together with Lemma~\ref{lem_vic} leads to 
that Bob wins the game on $T_n$ when $n\neq 2^m-3$.
\end{proof}

Theorem~\ref{thm_trigrid_bob_wins} shows that 
Bob can win the game when the starting vertex is $v_0^0$.
This is a common case, that is,
we can see that every even kernel graphs in $T_n$ must include $v_0^0$.

\begin{lem}\label{lem_must_include}
There is no even kernel graph on $T_n$ that does not include $v_0^0$ when $n>1$.
\end{lem}
\begin{proof}
We prove the lemma by contradiction and induction 
on the distance of $v^0_0$ and another vertex.
If an even kernel graph $H$ does not include $v_0^0$, 
then neither $v^1_0$ nor $v^1_1$ is contained in $H$ by the definition.
Therefore, all vertices whose distance from $v_0^0$ is $1$ must not be in $H$.

Assume that no vertex whose distance from $v_0^0$ is at most $k$ is 
in an even kernel graph $H$ with bipartite sets $B$ and $W$, 
where $B$ contains a starting vertex,
we can see that any $v_i^{k+1} (0 \leq i \leq k+1)$ cannot be in $B$ by definition; 
because any $v^k_j (0 \leq j \leq k)$ is not a member of $W$ from the assumption.
If $v^{k+1}_i$ is a member of $W$, 
by definition, $v^{k+1}_i$ must have two or four edges in $H$. 
This condition and local restrictions show that both $v^{k+2}_i$ and $v^{k+2}_{i+1}$ are a member of $B$.
This violates the definition for $B$. Therefore, $v^{k+1}_i$ cannot be a member of $W$.

By induction on $k$, 
if $H$ does not include $v_0^0$, any vertex is not a member of $H$, a contradiction.
Therefore, all even kernel graphs of $T_n$ must includes $v_0^0$.
\end{proof}

For the case when $n = 2^m - 3$ with $m\geq 2$, 
we have checked that Alice wins the game on $T_n$ with a starting vertex $v^0_0$ 
for small cases $n=1,5$. 
Furthermore, 
we confirm that there exists no even kernel graph of $T_n$ if $n = 2^m - 3$ with $m\geq 2$, as follows.

\begin{thm}\label{thm_trigrid_no_victory}
If $n = 2^m - 3$ with $m\geq 2$, 
then there exists no even kernel graph of the triangular grid graph $T_n$. 
\end{thm}
\begin{proof}
Suppose to the contrary that $T_n$ has a even kernel graph $H_n$.
From Lemma~\ref{lem_must_include}, any even kernel graph $H_n$ of $T_n$ 
contains $v_0^0$, $v^n_0$ and $v_n^n$.
Furthermore, these three vertices are in $B \subset V(H_n)$,
which is a subset containing a starting vertex.

By symmetry, 
let $i$ be the smallest number such that $v^{2i}_0 \notin B$,
i.e., if $v^{2j}_{2j} \notin B$ with $j < i$,
then we relabel $v^{k}_{0},v^{k}_{1},\dots,v^{k}_{k}$ as $v^{k}_{k},v^{k}_{k-1},\dots,v^{k}_{0}$
for any $k \in \{1,2,\dots,n\}$.
Then $v^1_0, \dots , v^{2i-1}_0$ are in $W$.
By definition and local restrictions, 
it must be hold that $v_{j}^{k} \in B$ when $j,k$ is even, otherwise $v_{j}^{k} \in W$.
We say this pattern on $\bigtriangleup v_0^0 v^{2i-2}_0 v_{2i-2}^{2i-2}$ 
\textit{close-packed} or 
$\bigtriangleup v_0^0 v^{2i-2}_0 v_{2i-2}^{2i-2}$ is called 
\textit{a close-packed triangle},
where $\bigtriangleup abc$ $(a,b,c \in V(T_n))$ denotes 
a triangular grid graph $T_{p}$ for some $p \in \{0,1,\dots,n\}$
which is contained in $T_n$ as a subgraph.

Since $v^{2i-1}_0, v^{2i-1}_1$ and $v^{2i}_0$ are in $W$ and $v^{2i-2}_0 \in B$, 
we have $v^{2i}_1 \in B$, and this leads to $v^{2i+1}_1 \in W$ 
and $v^{2i+1}_0 \in B$.
Observe that $v^{2i}_{2i} \notin B$, since otherwise, 
$v^{2i}_{2}$ is also in $B$ by the observation in the second paragraph, 
which contradicts $v^{2i}_1 \in B$.
Moreover, with the similar observation, 
$v^{2i}_{2i-1}, v^{2i+1}_{2i+1} \in B$ and $v^{2i+1}_{2i} \in W$.

Two black vertices $v^{2i}_1$ and $v^{2i}_{2i-1}$ 
and white vertices $v^{2i-1}_j$, where $0 \leq j \leq 2i-1$, force that
all $v^{2i}_j \in B$ if $j$ is odd and all $v^{2i}_j \in W$ if $j$ is even. 
By local restrictions, 
$\bigtriangleup v^{2i}_{1} v^{2i}_{2i-1} v^{4i-2}_{2i-1}$ is close-packed;
note that $v^{2i+2}_{2},v^{2i+2}_{2i} \notin B$,
since if $v^{2i+2}_{2} \in B$ (resp., $v^{2i+2}_{2i} \in B$), 
then $v^{2i+1}_{1}$ (resp., $v^{2i+1}_{2i}$) in $W$ must be of degree~3,
a contradiction.
Furthermore, 
$\bigtriangleup v^{2i}_{1} v^{2i}_{2i-1} v^{4i-2}_{2i-1}$ forces 
$v^{2i+j}_j$, where $0 \leq j \leq 2i-1$ and $v^{2i+j}_{2i}$, 
where $0 \leq j \leq 2i-1$ are in $W$. 
Note that whether $v^{4i}_{2i}$ is in $B$ or $W$ 
is not revealed yet under above discussions.
(Since the degree of a vertex in $W$ may be zero, 
every vertex of $T_n$ can be a member in $V(H_n)$.)
Now we have:

\begin{enumerate}
\item White vertices $v^{2i+j}_j$, where $0 \leq j \leq 2i-1$, 
and black one $v^{2i+1}_{0}$ generate 
a close-packed triangle $\bigtriangleup v^{2i+1}_{0} v^{4i-1}_{0} v^{4i-1}_{2i-2}$.
\item White vertices $v^{2i+j}_{2i}$, where $0 \leq j \leq 2i-1$, 
and black one $v^{2i+1}_{2i+1}$ generate
a close-packed triangle $\bigtriangleup v^{2i+1}_{2i+1} v^{4i-1}_{2i+1} v^{4i-1}_{4i-1}$.
\end{enumerate}

These successive generation can stop if $n = 4i-1$.
If $n = 4i-1$, $H_n$ is constructed by four close-packed triangles.
If $n < 4i-1$, above generation are not satisfied. 
Therefore there does not exist such $i$ under that $n$.
Otherwise, $n > 4i-1$, above generation must continue, 
as follows:

Two close-packed triangles $\bigtriangleup v^{2i+1}_{0} v^{4i-1}_{0} v^{4i-1}_{2i-2}$ 
and $\bigtriangleup v^{2i+1}_{2i+1} v^{4i-1}_{2i+1} v^{4i-1}_{4i-1}$ 
force that $v^{4i}_{j} \in W$, where $0 \leq j \leq 4i$ and $j \neq 2i$.
We next focus on the fact that 
$v^{4i-1}_{2i-1}, v^{4i-1}_{2i}, v^{4i}_{2i-1}$ and $v^{4i}_{2i+1}$ must be in $W$.
This fact implies $v^{4i}_{2i}$ cannot be a member of $H_n$
since local constrains force $v^{4i+1}_{2i-1}, v^{4i+1}_{2i+2}, v^{4i+2}_{2i+1} \in B$.
These new black vertices generate new three close-packed triangles
$\bigtriangleup v^{4i+1}_{1} v^{4i+1}_{2i-1} v^{6i-1}_{2i-1},
 \bigtriangleup v^{4i+2}_{2i+1} v^{6i}_{2i+1} v^{6i}_{4i-1}$ and
$\bigtriangleup v^{4i+1}_{2i+2} v^{4i+1}_{4i} v^{6i-1}_{4i}$, 
and these new close-packed triangles force two extra close-packed triangles
$\bigtriangleup v^{4i+2}_{0} v^{6i}_{0} v^{6i}_{2i-2},
 \bigtriangleup v^{4i+2}_{4i+2} v^{6i}_{4i+2} v^{6i}_{6i}$.
In this case, 
these successive generation can stop if $n = 6i$,
and also this discussion can continue recursively if $n > 6i$.

Let $r$ be the number of recursion on the above discussion, 
i.e., $H_n$ contains $r^2$ close-packed triangles.
By the hypothesis,
there can exist such $i$ on $T_n$ if $n = r(2i+1)-3$, where $i,r \geq 1$,
which 
implies that only if $n$ can be represented as $n = r(2i+1)-3$, 
where $i,r \geq 1$, $H_n$ can exist.
Therefore, by the assumption that $n=2^m-3$,
there must not exist an even kernel graph for $T_n$ when $m>1$ 
since $2^{m-j}$ cannot be represented as $2i+1$ for any $i \geq 1$ and $j \leq m$,
a contradiction.
\end{proof}

Thus we propose the following conjecture which implies that
for every triangular grid graph $T_n$ with a starting vertex $v^0_0$,
Bob wins the game on $T_n$ if and only if $T_n$ contains an even kernel graph
with respect to $v^0_0$.

\begin{conj}\label{conj_trigrid}
If $n = 2^m - 3$ with $m\geq 2$, 
then Alice wins the game on the triangular grid graph $T_n$
with a starting vertex $v^0_0$.
\end{conj}

\section{Toroidal grid graphs}\label{sec_toro_grid}

In this section,
we investigate the feedback game on toroidal grid graphs.
The undirected edge geography on a grid graph 
(which is the Cartesian product of two paths)
is completely solved~\cite{edgegeo},
and the directed edge geography on a directed toroidal grid graph
is also investigated in~\cite{HH2003}.

\begin{df}[Toroidal grid graph]
A {\em toroidal grid graph} $Q(m,n)$ is the Cartesian product of 
two cycles $C_m = u_0u_1 \dots u_{m-1}$ and $C_n = v_0v_1 \dots v_{n-1}$
with $m\geq 2$ and $n\geq 2$,
that is, 
\begin{itemize}
\item $V(Q(m,n)) = \{(u_i,v_j) : i \in \{0,1,\dots,m-1\},\, j \in \{0,1,\dots,n-1\} \}$. 
\item $(u_i,v_j)(u_{i'},v_{j'}) \in E(Q(m,n))$ if and only if 
\begin{itemize}
\item $i=i'$ and $v_jv_{j'} \in E(C_n)$ or
\item $j=j'$ and $u_iu_{i'} \in E(C_m)$.
\end{itemize}
\end{itemize}
\end{df}

In other words, $Q(m,n)$ is a 4-regular {\em quadrangulation} embedded on the torus,
which is a graph on a surface with each face quadrangular.
For example, see Figure~\ref{torogrid};
by identifying the top and bottom (resp., right and left) sides 
along the direction of arrows, we have the toroidal grid graph $Q(3,4)$.
Note that $Q(m,n)$ is vertex-transitive, that is, 
there exists an automorphism of the graph mapping a vertex into any other vertex.
Thus the feedback game on $Q(m,n)$ does not depend on the choice of a starting vertex,
and hence,
toroidal grid graphs seem to be a reasonable subclass of 
connected Eulerian graphs with maximum degree at most~4
for considering the feedback game.

\begin{figure}[htb]
\centering
\unitlength 0.1in
\begin{picture}( 30.5000, 18.8000)(  6.6000,-27.8000)
%
\special{pn 8}%
\special{sh 1.000}%
\special{ar 2006 1406 50 50  0.0000000 6.2831853}%
%
\special{pn 8}%
\special{sh 1.000}%
\special{ar 2006 1806 50 50  0.0000000 6.2831853}%
%
\special{pn 8}%
\special{sh 1.000}%
\special{ar 2006 2206 50 50  0.0000000 6.2831853}%
%
\special{pn 8}%
\special{sh 1.000}%
\special{ar 2406 1406 50 50  0.0000000 6.2831853}%
%
\special{pn 8}%
\special{sh 1.000}%
\special{ar 2406 1806 50 50  0.0000000 6.2831853}%
%
\special{pn 8}%
\special{sh 1.000}%
\special{ar 2406 2206 50 50  0.0000000 6.2831853}%
%
\special{pn 8}%
\special{sh 1.000}%
\special{ar 2806 2206 50 50  0.0000000 6.2831853}%
%
\special{pn 8}%
\special{sh 1.000}%
\special{ar 2806 1806 50 50  0.0000000 6.2831853}%
%
\special{pn 8}%
\special{sh 1.000}%
\special{ar 2806 1406 50 50  0.0000000 6.2831853}%
%
\special{pn 8}%
\special{sh 1.000}%
\special{ar 3206 1406 50 50  0.0000000 6.2831853}%
%
\special{pn 8}%
\special{sh 1.000}%
\special{ar 3206 1806 50 50  0.0000000 6.2831853}%
%
\special{pn 8}%
\special{sh 1.000}%
\special{ar 3206 2206 50 50  0.0000000 6.2831853}%
%
\special{pn 8}%
\special{pa 3606 2206}%
\special{pa 1606 2206}%
\special{fp}%
%
\special{pn 8}%
\special{pa 1606 1806}%
\special{pa 3606 1806}%
\special{fp}%
%
\special{pn 8}%
\special{pa 3606 1406}%
\special{pa 1606 1406}%
\special{fp}%
%
\special{pn 8}%
\special{pa 2006 1006}%
\special{pa 2006 2606}%
\special{fp}%
%
\special{pn 8}%
\special{pa 2406 2606}%
\special{pa 2406 1006}%
\special{fp}%
%
\special{pn 8}%
\special{pa 2806 1006}%
\special{pa 2806 2606}%
\special{fp}%
%
\special{pn 8}%
\special{pa 3206 2606}%
\special{pa 3206 1006}%
\special{fp}%
%
\special{pn 8}%
\special{pa 3606 1006}%
\special{pa 1606 1006}%
\special{da 0.070}%
%
\special{pn 8}%
\special{pa 1606 1006}%
\special{pa 1606 2606}%
\special{da 0.070}%
%
\special{pn 8}%
\special{pa 1606 2606}%
\special{pa 3606 2606}%
\special{da 0.070}%
%
\special{pn 8}%
\special{pa 3606 2606}%
\special{pa 3606 1006}%
\special{da 0.070}%
%
\special{pn 13}%
\special{pa 1506 2606}%
\special{pa 1506 1006}%
\special{fp}%
\special{sh 1}%
\special{pa 1506 1006}%
\special{pa 1486 1072}%
\special{pa 1506 1058}%
\special{pa 1526 1072}%
\special{pa 1506 1006}%
\special{fp}%
\special{pa 1506 1006}%
\special{pa 1506 1006}%
\special{fp}%
%
\special{pn 13}%
\special{pa 3706 2606}%
\special{pa 3706 1006}%
\special{fp}%
\special{sh 1}%
\special{pa 3706 1006}%
\special{pa 3686 1072}%
\special{pa 3706 1058}%
\special{pa 3726 1072}%
\special{pa 3706 1006}%
\special{fp}%
%
\special{pn 13}%
\special{pa 1606 906}%
\special{pa 3606 906}%
\special{da 0.070}%
\special{sh 1}%
\special{pa 3606 906}%
\special{pa 3538 886}%
\special{pa 3552 906}%
\special{pa 3538 926}%
\special{pa 3606 906}%
\special{fp}%
\put(11.5500,-16.6500){\makebox(0,0){$(u_1,v_0)$}}%
\put(11.5500,-20.6500){\makebox(0,0){$(u_2,v_0)$}}%
\put(22.2500,-28.6500){\makebox(0,0){$(u_2,v_1)$}}%
%
\special{pn 13}%
\special{pa 1606 2706}%
\special{pa 3606 2706}%
\special{da 0.070}%
\special{sh 1}%
\special{pa 3606 2706}%
\special{pa 3538 2686}%
\special{pa 3552 2706}%
\special{pa 3538 2726}%
\special{pa 3606 2706}%
\special{fp}%
\put(11.6000,-12.7000){\makebox(0,0){$(u_0,v_0)$}}%
%
\special{pn 8}%
\special{pa 1440 1300}%
\special{pa 1910 1360}%
\special{fp}%
\special{sh 1}%
\special{pa 1910 1360}%
\special{pa 1846 1332}%
\special{pa 1858 1354}%
\special{pa 1842 1372}%
\special{pa 1910 1360}%
\special{fp}%
%
\special{pn 8}%
\special{pa 1440 1700}%
\special{pa 1910 1760}%
\special{fp}%
\special{sh 1}%
\special{pa 1910 1760}%
\special{pa 1846 1732}%
\special{pa 1858 1754}%
\special{pa 1842 1772}%
\special{pa 1910 1760}%
\special{fp}%
%
\special{pn 8}%
\special{pa 1440 2100}%
\special{pa 1910 2160}%
\special{fp}%
\special{sh 1}%
\special{pa 1910 2160}%
\special{pa 1846 2132}%
\special{pa 1858 2154}%
\special{pa 1842 2172}%
\special{pa 1910 2160}%
\special{fp}%
\end{picture}%
\caption{The toroidal grid graph $Q(3,4)$}
\label{torogrid}
\end{figure}
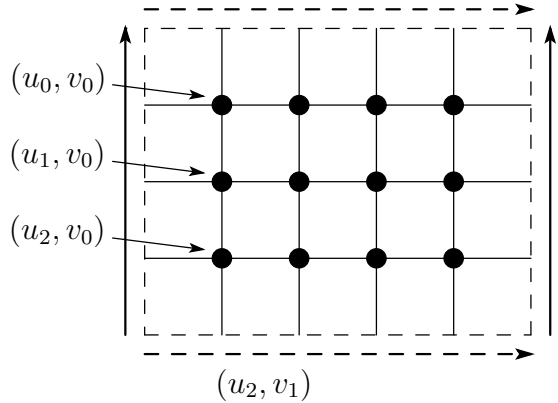

For several combinations of $m$ and $n$, 
we have determined a winner of the game, as follows.
In particular, if the greatest common divisor of $m$ and $n$,
denoted by $\gcd(m,n)$, is bigger than one, then Bob can win the game on $Q(m,n)$,
and otherwise it seems to be that Alice can win the game.

\begin{thm}\label{thm_torogrid_bob}
If $\gcd(m,n)=c > 1$, then Bob can win the game on $Q(m,n)$.
\end{thm}
\begin{proof}
By the assumption, let $m=ck$ and $n=ck'$ for some positive integers $k$ and $k'$.
The toroidal grid graph $Q(c,c)$ with a starting vertex $s = (u_0,v_0)$
has an even kernel graph $H^c$ with bipartite sets $B$ and $W$ such that
$(u_i,v_i) \in B$, $(u_i,v_{i+1}), (u_{i+1},v_i) \in W$ 
and edges $(u_i,v_i)(u_i,v_{i+1}),(u_i,v_i)(u_{i+1},v_i), 
(u_i,v_{i+1})(u_{i+1},v_{i+1})$ and $(u_{i+1},v_i)(u_{i+1},v_{i+1})$ are in $E(H^c)$
for any $i \in \{0,1,\dots,c-1\}$, where subscripts are modulo $c$ 
(see Figure~\ref{vics}).

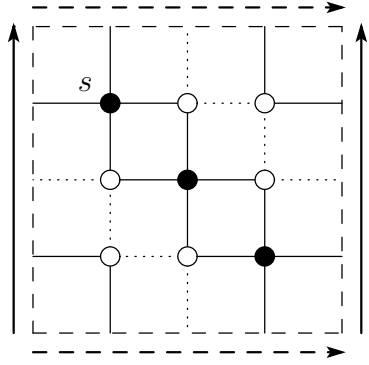
\begin{figure}[htb]
\centering
\unitlength 0.1in
\begin{picture}( 18.1000, 18.1000)(  6.9500,-21.0500)
%
\special{pn 8}%
\special{sh 1.000}%
\special{ar 1200 800 50 50  0.0000000 6.2831853}%
%
\special{pn 8}%
\special{sh 1.000}%
\special{ar 1600 1200 50 50  0.0000000 6.2831853}%
%
\special{pn 8}%
\special{sh 1.000}%
\special{ar 2000 1600 50 50  0.0000000 6.2831853}%
%
\special{pn 8}%
\special{pa 2400 400}%
\special{pa 800 400}%
\special{da 0.070}%
%
\special{pn 8}%
\special{pa 800 400}%
\special{pa 800 2000}%
\special{da 0.070}%
%
\special{pn 8}%
\special{pa 800 2000}%
\special{pa 2400 2000}%
\special{da 0.070}%
%
\special{pn 8}%
\special{pa 2400 2000}%
\special{pa 2400 400}%
\special{da 0.070}%
%
\special{pn 13}%
\special{pa 700 2000}%
\special{pa 700 400}%
\special{fp}%
\special{sh 1}%
\special{pa 700 400}%
\special{pa 680 468}%
\special{pa 700 454}%
\special{pa 720 468}%
\special{pa 700 400}%
\special{fp}%
\special{pa 700 400}%
\special{pa 700 400}%
\special{fp}%
%
\special{pn 13}%
\special{pa 2500 2000}%
\special{pa 2500 400}%
\special{fp}%
\special{sh 1}%
\special{pa 2500 400}%
\special{pa 2480 468}%
\special{pa 2500 454}%
\special{pa 2520 468}%
\special{pa 2500 400}%
\special{fp}%
%
\special{pn 13}%
\special{pa 800 300}%
\special{pa 2400 300}%
\special{da 0.070}%
\special{sh 1}%
\special{pa 2400 300}%
\special{pa 2334 280}%
\special{pa 2348 300}%
\special{pa 2334 320}%
\special{pa 2400 300}%
\special{fp}%
%
\special{pn 13}%
\special{pa 800 2100}%
\special{pa 2400 2100}%
\special{da 0.070}%
\special{sh 1}%
\special{pa 2400 2100}%
\special{pa 2334 2080}%
\special{pa 2348 2100}%
\special{pa 2334 2120}%
\special{pa 2400 2100}%
\special{fp}%
%
\special{pn 8}%
\special{pa 1200 800}%
\special{pa 1200 1200}%
\special{fp}%
%
\special{pn 8}%
\special{pa 1200 1200}%
\special{pa 1600 1200}%
\special{fp}%
%
\special{pn 8}%
\special{pa 1600 1200}%
\special{pa 1600 1600}%
\special{fp}%
%
\special{pn 8}%
\special{pa 1600 1600}%
\special{pa 2000 1600}%
\special{fp}%
%
\special{pn 8}%
\special{pa 2000 1600}%
\special{pa 2000 1200}%
\special{fp}%
%
\special{pn 8}%
\special{pa 2000 1200}%
\special{pa 1600 1200}%
\special{fp}%
%
\special{pn 8}%
\special{pa 1600 1200}%
\special{pa 1600 800}%
\special{fp}%
%
\special{pn 8}%
\special{pa 1600 800}%
\special{pa 1200 800}%
\special{fp}%
%
\special{pn 8}%
\special{pa 1200 800}%
\special{pa 1200 400}%
\special{fp}%
%
\special{pn 8}%
\special{pa 1200 800}%
\special{pa 800 800}%
\special{fp}%
%
\special{pn 8}%
\special{pa 2000 1600}%
\special{pa 2000 2000}%
\special{fp}%
%
\special{pn 8}%
\special{pa 2000 1600}%
\special{pa 2400 1600}%
\special{fp}%
%
\special{pn 8}%
\special{pa 1200 1600}%
\special{pa 800 1600}%
\special{fp}%
%
\special{pn 8}%
\special{pa 1200 1600}%
\special{pa 1200 2000}%
\special{fp}%
%
\special{pn 8}%
\special{pa 2400 800}%
\special{pa 2000 800}%
\special{fp}%
%
\special{pn 8}%
\special{pa 2000 800}%
\special{pa 2000 400}%
\special{fp}%
%
\special{pn 8}%
\special{pa 1600 1600}%
\special{pa 1200 1600}%
\special{dt 0.045}%
%
\special{pn 8}%
\special{pa 1200 1600}%
\special{pa 1200 1200}%
\special{dt 0.045}%
%
\special{pn 8}%
\special{pa 1200 1200}%
\special{pa 800 1200}%
\special{dt 0.045}%
%
\special{pn 8}%
\special{pa 1600 800}%
\special{pa 2000 800}%
\special{dt 0.045}%
%
\special{pn 8}%
\special{pa 1600 800}%
\special{pa 1600 400}%
\special{dt 0.045}%
%
\special{pn 8}%
\special{pa 1600 1600}%
\special{pa 1600 2000}%
\special{dt 0.045}%
%
\special{pn 8}%
\special{pa 2000 1200}%
\special{pa 2000 800}%
\special{dt 0.045}%
%
\special{pn 8}%
\special{pa 2000 1200}%
\special{pa 2400 1200}%
\special{dt 0.045}%
%
\special{pn 8}%
\special{sh 0}%
\special{ar 1600 1600 50 50  0.0000000 6.2831853}%
%
\special{pn 8}%
\special{sh 0}%
\special{ar 1200 1600 50 50  0.0000000 6.2831853}%
%
\special{pn 8}%
\special{sh 0}%
\special{ar 1200 1200 50 50  0.0000000 6.2831853}%
%
\special{pn 8}%
\special{sh 0}%
\special{ar 1600 800 50 50  0.0000000 6.2831853}%
%
\special{pn 8}%
\special{sh 0}%
\special{ar 2000 800 50 50  0.0000000 6.2831853}%
%
\special{pn 8}%
\special{sh 0}%
\special{ar 2000 1200 50 50  0.0000000 6.2831853}%
\put(10.7000,-7.0000){\makebox(0,0){$s$}}%
\end{picture}%
\caption{An even kernel graph of $Q(3,3)$}
\label{vics}
\end{figure}

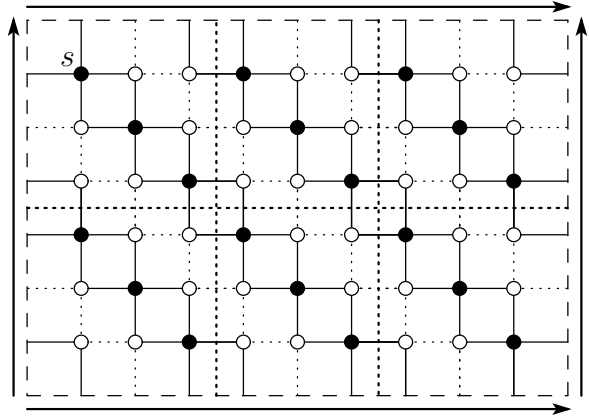
\begin{figure}[htb]
\centering
\unitlength 0.1in
\begin{picture}( 29.5000, 21.1000)(  7.2500,-24.3500)
%
\special{pn 8}%
\special{sh 1.000}%
\special{ar 1080 680 36 36  0.0000000 6.2831853}%
%
\special{pn 8}%
\special{sh 1.000}%
\special{ar 1360 960 36 36  0.0000000 6.2831853}%
%
\special{pn 8}%
\special{sh 1.000}%
\special{ar 1640 1240 36 36  0.0000000 6.2831853}%
%
\special{pn 8}%
\special{pa 1080 680}%
\special{pa 1080 960}%
\special{fp}%
%
\special{pn 8}%
\special{pa 1080 960}%
\special{pa 1360 960}%
\special{fp}%
%
\special{pn 8}%
\special{pa 1360 960}%
\special{pa 1360 1240}%
\special{fp}%
%
\special{pn 8}%
\special{pa 1360 1240}%
\special{pa 1640 1240}%
\special{fp}%
%
\special{pn 8}%
\special{pa 1640 1240}%
\special{pa 1640 960}%
\special{fp}%
%
\special{pn 8}%
\special{pa 1640 960}%
\special{pa 1360 960}%
\special{fp}%
%
\special{pn 8}%
\special{pa 1360 960}%
\special{pa 1360 680}%
\special{fp}%
%
\special{pn 8}%
\special{pa 1360 680}%
\special{pa 1080 680}%
\special{fp}%
%
\special{pn 8}%
\special{pa 1080 680}%
\special{pa 1080 400}%
\special{fp}%
%
\special{pn 8}%
\special{pa 1080 680}%
\special{pa 800 680}%
\special{fp}%
%
\special{pn 8}%
\special{pa 1640 1240}%
\special{pa 1640 1520}%
\special{fp}%
%
\special{pn 8}%
\special{pa 1640 1240}%
\special{pa 1920 1240}%
\special{fp}%
%
\special{pn 8}%
\special{pa 1080 1240}%
\special{pa 800 1240}%
\special{fp}%
%
\special{pn 8}%
\special{pa 1080 1240}%
\special{pa 1080 1520}%
\special{fp}%
%
\special{pn 8}%
\special{pa 1920 680}%
\special{pa 1640 680}%
\special{fp}%
%
\special{pn 8}%
\special{pa 1640 680}%
\special{pa 1640 400}%
\special{fp}%
%
\special{pn 8}%
\special{pa 1360 1240}%
\special{pa 1080 1240}%
\special{dt 0.045}%
%
\special{pn 8}%
\special{pa 1080 1240}%
\special{pa 1080 960}%
\special{dt 0.045}%
%
\special{pn 8}%
\special{pa 1080 960}%
\special{pa 800 960}%
\special{dt 0.045}%
%
\special{pn 8}%
\special{pa 1360 680}%
\special{pa 1640 680}%
\special{dt 0.045}%
%
\special{pn 8}%
\special{pa 1360 680}%
\special{pa 1360 400}%
\special{dt 0.045}%
%
\special{pn 8}%
\special{pa 1360 1240}%
\special{pa 1360 1520}%
\special{dt 0.045}%
%
\special{pn 8}%
\special{pa 1640 960}%
\special{pa 1640 680}%
\special{dt 0.045}%
%
\special{pn 8}%
\special{pa 1640 960}%
\special{pa 1920 960}%
\special{dt 0.045}%
%
\special{pn 8}%
\special{sh 0}%
\special{ar 1080 960 36 36  0.0000000 6.2831853}%
%
\special{pn 8}%
\special{sh 0}%
\special{ar 1360 680 36 36  0.0000000 6.2831853}%
%
\special{pn 8}%
\special{sh 1.000}%
\special{ar 1080 1520 36 36  0.0000000 6.2831853}%
%
\special{pn 8}%
\special{sh 1.000}%
\special{ar 1360 1800 36 36  0.0000000 6.2831853}%
%
\special{pn 8}%
\special{sh 1.000}%
\special{ar 1640 2080 36 36  0.0000000 6.2831853}%
%
\special{pn 8}%
\special{pa 1080 1520}%
\special{pa 1080 1800}%
\special{fp}%
%
\special{pn 8}%
\special{pa 1080 1800}%
\special{pa 1360 1800}%
\special{fp}%
%
\special{pn 8}%
\special{pa 1360 1800}%
\special{pa 1360 2080}%
\special{fp}%
%
\special{pn 8}%
\special{pa 1360 2080}%
\special{pa 1640 2080}%
\special{fp}%
%
\special{pn 8}%
\special{pa 1640 2080}%
\special{pa 1640 1800}%
\special{fp}%
%
\special{pn 8}%
\special{pa 1640 1800}%
\special{pa 1360 1800}%
\special{fp}%
%
\special{pn 8}%
\special{pa 1360 1800}%
\special{pa 1360 1520}%
\special{fp}%
%
\special{pn 8}%
\special{pa 1360 1520}%
\special{pa 1080 1520}%
\special{fp}%
%
\special{pn 8}%
\special{pa 1080 1520}%
\special{pa 1080 1240}%
\special{fp}%
%
\special{pn 8}%
\special{pa 1080 1520}%
\special{pa 800 1520}%
\special{fp}%
%
\special{pn 8}%
\special{pa 1640 2080}%
\special{pa 1640 2360}%
\special{fp}%
%
\special{pn 8}%
\special{pa 1640 2080}%
\special{pa 1920 2080}%
\special{fp}%
%
\special{pn 8}%
\special{pa 1080 2080}%
\special{pa 800 2080}%
\special{fp}%
%
\special{pn 8}%
\special{pa 1080 2080}%
\special{pa 1080 2360}%
\special{fp}%
%
\special{pn 8}%
\special{pa 1920 1520}%
\special{pa 1640 1520}%
\special{fp}%
%
\special{pn 8}%
\special{pa 1640 1520}%
\special{pa 1640 1240}%
\special{fp}%
%
\special{pn 8}%
\special{pa 1360 2080}%
\special{pa 1080 2080}%
\special{dt 0.045}%
%
\special{pn 8}%
\special{pa 1080 2080}%
\special{pa 1080 1800}%
\special{dt 0.045}%
%
\special{pn 8}%
\special{pa 1080 1800}%
\special{pa 800 1800}%
\special{dt 0.045}%
%
\special{pn 8}%
\special{pa 1360 1520}%
\special{pa 1640 1520}%
\special{dt 0.045}%
%
\special{pn 8}%
\special{pa 1360 1520}%
\special{pa 1360 1240}%
\special{dt 0.045}%
%
\special{pn 8}%
\special{pa 1360 2080}%
\special{pa 1360 2360}%
\special{dt 0.045}%
%
\special{pn 8}%
\special{pa 1640 1800}%
\special{pa 1640 1520}%
\special{dt 0.045}%
%
\special{pn 8}%
\special{pa 1640 1800}%
\special{pa 1920 1800}%
\special{dt 0.045}%
%
\special{pn 8}%
\special{sh 0}%
\special{ar 1360 2080 36 36  0.0000000 6.2831853}%
%
\special{pn 8}%
\special{sh 0}%
\special{ar 1080 2080 36 36  0.0000000 6.2831853}%
%
\special{pn 8}%
\special{sh 0}%
\special{ar 1080 1800 36 36  0.0000000 6.2831853}%
%
\special{pn 8}%
\special{sh 0}%
\special{ar 1360 1520 36 36  0.0000000 6.2831853}%
%
\special{pn 8}%
\special{sh 1.000}%
\special{ar 1920 1520 36 36  0.0000000 6.2831853}%
%
\special{pn 8}%
\special{sh 1.000}%
\special{ar 2200 1800 36 36  0.0000000 6.2831853}%
%
\special{pn 8}%
\special{sh 1.000}%
\special{ar 2480 2080 36 36  0.0000000 6.2831853}%
%
\special{pn 8}%
\special{pa 1920 1520}%
\special{pa 1920 1800}%
\special{fp}%
%
\special{pn 8}%
\special{pa 1920 1800}%
\special{pa 2200 1800}%
\special{fp}%
%
\special{pn 8}%
\special{pa 2200 1800}%
\special{pa 2200 2080}%
\special{fp}%
%
\special{pn 8}%
\special{pa 2200 2080}%
\special{pa 2480 2080}%
\special{fp}%
%
\special{pn 8}%
\special{pa 2480 2080}%
\special{pa 2480 1800}%
\special{fp}%
%
\special{pn 8}%
\special{pa 2480 1800}%
\special{pa 2200 1800}%
\special{fp}%
%
\special{pn 8}%
\special{pa 2200 1800}%
\special{pa 2200 1520}%
\special{fp}%
%
\special{pn 8}%
\special{pa 2200 1520}%
\special{pa 1920 1520}%
\special{fp}%
%
\special{pn 8}%
\special{pa 1920 1520}%
\special{pa 1920 1240}%
\special{fp}%
%
\special{pn 8}%
\special{pa 1920 1520}%
\special{pa 1640 1520}%
\special{fp}%
%
\special{pn 8}%
\special{pa 2480 2080}%
\special{pa 2480 2360}%
\special{fp}%
%
\special{pn 8}%
\special{pa 2480 2080}%
\special{pa 2760 2080}%
\special{fp}%
%
\special{pn 8}%
\special{pa 1920 2080}%
\special{pa 1640 2080}%
\special{fp}%
%
\special{pn 8}%
\special{pa 1920 2080}%
\special{pa 1920 2360}%
\special{fp}%
%
\special{pn 8}%
\special{pa 2760 1520}%
\special{pa 2480 1520}%
\special{fp}%
%
\special{pn 8}%
\special{pa 2480 1520}%
\special{pa 2480 1240}%
\special{fp}%
%
\special{pn 8}%
\special{pa 2200 2080}%
\special{pa 1920 2080}%
\special{dt 0.045}%
%
\special{pn 8}%
\special{pa 1920 2080}%
\special{pa 1920 1800}%
\special{dt 0.045}%
%
\special{pn 8}%
\special{pa 1920 1800}%
\special{pa 1640 1800}%
\special{dt 0.045}%
%
\special{pn 8}%
\special{pa 2200 1520}%
\special{pa 2480 1520}%
\special{dt 0.045}%
%
\special{pn 8}%
\special{pa 2200 1520}%
\special{pa 2200 1240}%
\special{dt 0.045}%
%
\special{pn 8}%
\special{pa 2200 2080}%
\special{pa 2200 2360}%
\special{dt 0.045}%
%
\special{pn 8}%
\special{pa 2480 1800}%
\special{pa 2480 1520}%
\special{dt 0.045}%
%
\special{pn 8}%
\special{pa 2480 1800}%
\special{pa 2760 1800}%
\special{dt 0.045}%
%
\special{pn 8}%
\special{sh 0}%
\special{ar 2200 2080 36 36  0.0000000 6.2831853}%
%
\special{pn 8}%
\special{sh 0}%
\special{ar 1920 2080 36 36  0.0000000 6.2831853}%
%
\special{pn 8}%
\special{sh 0}%
\special{ar 1920 1800 36 36  0.0000000 6.2831853}%
%
\special{pn 8}%
\special{sh 1.000}%
\special{ar 1920 680 36 36  0.0000000 6.2831853}%
%
\special{pn 8}%
\special{sh 1.000}%
\special{ar 2200 960 36 36  0.0000000 6.2831853}%
%
\special{pn 8}%
\special{sh 1.000}%
\special{ar 2480 1240 36 36  0.0000000 6.2831853}%
%
\special{pn 8}%
\special{pa 1920 680}%
\special{pa 1920 960}%
\special{fp}%
%
\special{pn 8}%
\special{pa 1920 960}%
\special{pa 2200 960}%
\special{fp}%
%
\special{pn 8}%
\special{pa 2200 960}%
\special{pa 2200 1240}%
\special{fp}%
%
\special{pn 8}%
\special{pa 2200 1240}%
\special{pa 2480 1240}%
\special{fp}%
%
\special{pn 8}%
\special{pa 2480 1240}%
\special{pa 2480 960}%
\special{fp}%
%
\special{pn 8}%
\special{pa 2480 960}%
\special{pa 2200 960}%
\special{fp}%
%
\special{pn 8}%
\special{pa 2200 960}%
\special{pa 2200 680}%
\special{fp}%
%
\special{pn 8}%
\special{pa 2200 680}%
\special{pa 1920 680}%
\special{fp}%
%
\special{pn 8}%
\special{pa 1920 680}%
\special{pa 1920 400}%
\special{fp}%
%
\special{pn 8}%
\special{pa 1920 680}%
\special{pa 1640 680}%
\special{fp}%
%
\special{pn 8}%
\special{pa 2480 1240}%
\special{pa 2480 1520}%
\special{fp}%
%
\special{pn 8}%
\special{pa 2480 1240}%
\special{pa 2760 1240}%
\special{fp}%
%
\special{pn 8}%
\special{pa 1920 1240}%
\special{pa 1640 1240}%
\special{fp}%
%
\special{pn 8}%
\special{pa 1920 1240}%
\special{pa 1920 1520}%
\special{fp}%
%
\special{pn 8}%
\special{pa 2760 680}%
\special{pa 2480 680}%
\special{fp}%
%
\special{pn 8}%
\special{pa 2480 680}%
\special{pa 2480 400}%
\special{fp}%
%
\special{pn 8}%
\special{pa 2200 1240}%
\special{pa 1920 1240}%
\special{dt 0.045}%
%
\special{pn 8}%
\special{pa 1920 1240}%
\special{pa 1920 960}%
\special{dt 0.045}%
%
\special{pn 8}%
\special{pa 1920 960}%
\special{pa 1640 960}%
\special{dt 0.045}%
%
\special{pn 8}%
\special{pa 2200 680}%
\special{pa 2480 680}%
\special{dt 0.045}%
%
\special{pn 8}%
\special{pa 2200 680}%
\special{pa 2200 400}%
\special{dt 0.045}%
%
\special{pn 8}%
\special{pa 2200 1240}%
\special{pa 2200 1520}%
\special{dt 0.045}%
%
\special{pn 8}%
\special{pa 2480 960}%
\special{pa 2480 680}%
\special{dt 0.045}%
%
\special{pn 8}%
\special{pa 2480 960}%
\special{pa 2760 960}%
\special{dt 0.045}%
%
\special{pn 8}%
\special{sh 0}%
\special{ar 2200 1240 36 36  0.0000000 6.2831853}%
%
\special{pn 8}%
\special{sh 0}%
\special{ar 1920 1240 36 36  0.0000000 6.2831853}%
%
\special{pn 8}%
\special{sh 0}%
\special{ar 1920 960 36 36  0.0000000 6.2831853}%
%
\special{pn 8}%
\special{sh 0}%
\special{ar 2200 680 36 36  0.0000000 6.2831853}%
%
\special{pn 8}%
\special{pa 2480 1240}%
\special{pa 2760 1240}%
\special{fp}%
%
\special{pn 8}%
\special{pa 2760 680}%
\special{pa 2480 680}%
\special{fp}%
%
\special{pn 8}%
\special{pa 2480 960}%
\special{pa 2760 960}%
\special{dt 0.045}%
%
\special{pn 8}%
\special{pa 2760 1520}%
\special{pa 2760 1240}%
\special{fp}%
%
\special{pn 8}%
\special{pa 3320 1520}%
\special{pa 3320 1240}%
\special{fp}%
%
\special{pn 8}%
\special{pa 3040 1520}%
\special{pa 3040 1240}%
\special{dt 0.045}%
%
\special{pn 8}%
\special{sh 1.000}%
\special{ar 2760 680 36 36  0.0000000 6.2831853}%
%
\special{pn 8}%
\special{sh 1.000}%
\special{ar 3040 960 36 36  0.0000000 6.2831853}%
%
\special{pn 8}%
\special{sh 1.000}%
\special{ar 3320 1240 36 36  0.0000000 6.2831853}%
%
\special{pn 8}%
\special{pa 2760 680}%
\special{pa 2760 960}%
\special{fp}%
%
\special{pn 8}%
\special{pa 2760 960}%
\special{pa 3040 960}%
\special{fp}%
%
\special{pn 8}%
\special{pa 3040 960}%
\special{pa 3040 1240}%
\special{fp}%
%
\special{pn 8}%
\special{pa 3040 1240}%
\special{pa 3320 1240}%
\special{fp}%
%
\special{pn 8}%
\special{pa 3320 1240}%
\special{pa 3320 960}%
\special{fp}%
%
\special{pn 8}%
\special{pa 3320 960}%
\special{pa 3040 960}%
\special{fp}%
%
\special{pn 8}%
\special{pa 3040 960}%
\special{pa 3040 680}%
\special{fp}%
%
\special{pn 8}%
\special{pa 3040 680}%
\special{pa 2760 680}%
\special{fp}%
%
\special{pn 8}%
\special{pa 2760 680}%
\special{pa 2760 400}%
\special{fp}%
%
\special{pn 8}%
\special{pa 2760 680}%
\special{pa 2480 680}%
\special{fp}%
%
\special{pn 8}%
\special{pa 3320 1240}%
\special{pa 3320 1520}%
\special{fp}%
%
\special{pn 8}%
\special{pa 3320 1240}%
\special{pa 3600 1240}%
\special{fp}%
%
\special{pn 8}%
\special{pa 2760 1240}%
\special{pa 2480 1240}%
\special{fp}%
%
\special{pn 8}%
\special{pa 2760 1240}%
\special{pa 2760 1520}%
\special{fp}%
%
\special{pn 8}%
\special{pa 3600 680}%
\special{pa 3320 680}%
\special{fp}%
%
\special{pn 8}%
\special{pa 3320 680}%
\special{pa 3320 400}%
\special{fp}%
%
\special{pn 8}%
\special{pa 3040 1240}%
\special{pa 2760 1240}%
\special{dt 0.045}%
%
\special{pn 8}%
\special{pa 2760 1240}%
\special{pa 2760 960}%
\special{dt 0.045}%
%
\special{pn 8}%
\special{pa 2760 960}%
\special{pa 2480 960}%
\special{dt 0.045}%
%
\special{pn 8}%
\special{pa 3040 680}%
\special{pa 3320 680}%
\special{dt 0.045}%
%
\special{pn 8}%
\special{pa 3040 680}%
\special{pa 3040 400}%
\special{dt 0.045}%
%
\special{pn 8}%
\special{pa 3040 1240}%
\special{pa 3040 1520}%
\special{dt 0.045}%
%
\special{pn 8}%
\special{pa 3320 960}%
\special{pa 3320 680}%
\special{dt 0.045}%
%
\special{pn 8}%
\special{pa 3320 960}%
\special{pa 3600 960}%
\special{dt 0.045}%
%
\special{pn 8}%
\special{sh 0}%
\special{ar 2760 960 36 36  0.0000000 6.2831853}%
%
\special{pn 8}%
\special{sh 0}%
\special{ar 3040 680 36 36  0.0000000 6.2831853}%
%
\special{pn 8}%
\special{sh 0}%
\special{ar 3320 680 36 36  0.0000000 6.2831853}%
%
\special{pn 8}%
\special{sh 0}%
\special{ar 3320 960 36 36  0.0000000 6.2831853}%
%
\special{pn 8}%
\special{pa 2480 2080}%
\special{pa 2760 2080}%
\special{fp}%
%
\special{pn 8}%
\special{pa 2760 1520}%
\special{pa 2480 1520}%
\special{fp}%
%
\special{pn 8}%
\special{pa 2480 1800}%
\special{pa 2760 1800}%
\special{dt 0.045}%
%
\special{pn 8}%
\special{pa 2760 2360}%
\special{pa 2760 2080}%
\special{fp}%
%
\special{pn 8}%
\special{pa 3320 2360}%
\special{pa 3320 2080}%
\special{fp}%
%
\special{pn 8}%
\special{pa 3040 2360}%
\special{pa 3040 2080}%
\special{dt 0.045}%
%
\special{pn 8}%
\special{sh 1.000}%
\special{ar 2760 1520 36 36  0.0000000 6.2831853}%
%
\special{pn 8}%
\special{sh 1.000}%
\special{ar 3040 1800 36 36  0.0000000 6.2831853}%
%
\special{pn 8}%
\special{sh 1.000}%
\special{ar 3320 2080 36 36  0.0000000 6.2831853}%
%
\special{pn 8}%
\special{pa 2760 1520}%
\special{pa 2760 1800}%
\special{fp}%
%
\special{pn 8}%
\special{pa 2760 1800}%
\special{pa 3040 1800}%
\special{fp}%
%
\special{pn 8}%
\special{pa 3040 1800}%
\special{pa 3040 2080}%
\special{fp}%
%
\special{pn 8}%
\special{pa 3040 2080}%
\special{pa 3320 2080}%
\special{fp}%
%
\special{pn 8}%
\special{pa 3320 2080}%
\special{pa 3320 1800}%
\special{fp}%
%
\special{pn 8}%
\special{pa 3320 1800}%
\special{pa 3040 1800}%
\special{fp}%
%
\special{pn 8}%
\special{pa 3040 1800}%
\special{pa 3040 1520}%
\special{fp}%
%
\special{pn 8}%
\special{pa 3040 1520}%
\special{pa 2760 1520}%
\special{fp}%
%
\special{pn 8}%
\special{pa 2760 1520}%
\special{pa 2760 1240}%
\special{fp}%
%
\special{pn 8}%
\special{pa 2760 1520}%
\special{pa 2480 1520}%
\special{fp}%
%
\special{pn 8}%
\special{pa 3320 2080}%
\special{pa 3320 2360}%
\special{fp}%
%
\special{pn 8}%
\special{pa 3320 2080}%
\special{pa 3600 2080}%
\special{fp}%
%
\special{pn 8}%
\special{pa 2760 2080}%
\special{pa 2480 2080}%
\special{fp}%
%
\special{pn 8}%
\special{pa 2760 2080}%
\special{pa 2760 2360}%
\special{fp}%
%
\special{pn 8}%
\special{pa 3600 1520}%
\special{pa 3320 1520}%
\special{fp}%
%
\special{pn 8}%
\special{pa 3320 1520}%
\special{pa 3320 1240}%
\special{fp}%
%
\special{pn 8}%
\special{pa 3040 2080}%
\special{pa 2760 2080}%
\special{dt 0.045}%
%
\special{pn 8}%
\special{pa 2760 2080}%
\special{pa 2760 1800}%
\special{dt 0.045}%
%
\special{pn 8}%
\special{pa 2760 1800}%
\special{pa 2480 1800}%
\special{dt 0.045}%
%
\special{pn 8}%
\special{pa 3040 1520}%
\special{pa 3320 1520}%
\special{dt 0.045}%
%
\special{pn 8}%
\special{pa 3040 1520}%
\special{pa 3040 1240}%
\special{dt 0.045}%
%
\special{pn 8}%
\special{pa 3040 2080}%
\special{pa 3040 2360}%
\special{dt 0.045}%
%
\special{pn 8}%
\special{pa 3320 1800}%
\special{pa 3320 1520}%
\special{dt 0.045}%
%
\special{pn 8}%
\special{pa 3320 1800}%
\special{pa 3600 1800}%
\special{dt 0.045}%
%
\special{pn 8}%
\special{sh 0}%
\special{ar 3040 2080 36 36  0.0000000 6.2831853}%
%
\special{pn 8}%
\special{sh 0}%
\special{ar 2760 2080 36 36  0.0000000 6.2831853}%
%
\special{pn 8}%
\special{sh 0}%
\special{ar 2760 1800 36 36  0.0000000 6.2831853}%
%
\special{pn 8}%
\special{sh 0}%
\special{ar 3040 1520 36 36  0.0000000 6.2831853}%
%
\special{pn 8}%
\special{sh 0}%
\special{ar 3320 1520 36 36  0.0000000 6.2831853}%
%
\special{pn 8}%
\special{sh 0}%
\special{ar 3320 1800 36 36  0.0000000 6.2831853}%
%
\special{pn 8}%
\special{sh 0}%
\special{ar 1640 1800 36 36  0.0000000 6.2831853}%
%
\special{pn 8}%
\special{sh 0}%
\special{ar 1640 1520 36 36  0.0000000 6.2831853}%
%
\special{pn 8}%
\special{sh 0}%
\special{ar 1360 1240 36 36  0.0000000 6.2831853}%
%
\special{pn 8}%
\special{sh 0}%
\special{ar 1080 1240 36 36  0.0000000 6.2831853}%
%
\special{pn 8}%
\special{sh 0}%
\special{ar 1640 960 36 36  0.0000000 6.2831853}%
%
\special{pn 8}%
\special{sh 0}%
\special{ar 1640 680 36 36  0.0000000 6.2831853}%
%
\special{pn 8}%
\special{sh 0}%
\special{ar 2200 1520 36 36  0.0000000 6.2831853}%
%
\special{pn 8}%
\special{sh 0}%
\special{ar 2480 1800 36 36  0.0000000 6.2831853}%
%
\special{pn 8}%
\special{sh 0}%
\special{ar 2480 1520 36 36  0.0000000 6.2831853}%
%
\special{pn 8}%
\special{sh 0}%
\special{ar 2760 1240 36 36  0.0000000 6.2831853}%
%
\special{pn 8}%
\special{sh 0}%
\special{ar 3040 1240 36 36  0.0000000 6.2831853}%
%
\special{pn 8}%
\special{sh 0}%
\special{ar 2480 680 36 36  0.0000000 6.2831853}%
%
\special{pn 8}%
\special{sh 0}%
\special{ar 2480 960 36 36  0.0000000 6.2831853}%
%
\special{pn 8}%
\special{pa 3600 400}%
\special{pa 800 400}%
\special{da 0.070}%
%
\special{pn 8}%
\special{pa 800 400}%
\special{pa 800 2360}%
\special{da 0.070}%
%
\special{pn 8}%
\special{pa 800 2360}%
\special{pa 3600 2360}%
\special{da 0.070}%
%
\special{pn 8}%
\special{pa 3600 2360}%
\special{pa 3600 400}%
\special{da 0.070}%
%
\special{pn 13}%
\special{pa 3670 2360}%
\special{pa 3670 400}%
\special{fp}%
\special{sh 1}%
\special{pa 3670 400}%
\special{pa 3650 468}%
\special{pa 3670 454}%
\special{pa 3690 468}%
\special{pa 3670 400}%
\special{fp}%
%
\special{pn 13}%
\special{pa 730 2360}%
\special{pa 730 400}%
\special{fp}%
\special{sh 1}%
\special{pa 730 400}%
\special{pa 710 468}%
\special{pa 730 454}%
\special{pa 750 468}%
\special{pa 730 400}%
\special{fp}%
%
\special{pn 13}%
\special{pa 800 330}%
\special{pa 3600 330}%
\special{fp}%
\special{sh 1}%
\special{pa 3600 330}%
\special{pa 3534 310}%
\special{pa 3548 330}%
\special{pa 3534 350}%
\special{pa 3600 330}%
\special{fp}%
%
\special{pn 13}%
\special{pa 800 2430}%
\special{pa 3600 2430}%
\special{fp}%
\special{sh 1}%
\special{pa 3600 2430}%
\special{pa 3534 2410}%
\special{pa 3548 2430}%
\special{pa 3534 2450}%
\special{pa 3600 2430}%
\special{fp}%
\put(10.1000,-6.0300){\makebox(0,0){$s$}}%
%
\special{pn 13}%
\special{pa 800 1380}%
\special{pa 3600 1380}%
\special{dt 0.045}%
%
\special{pn 13}%
\special{pa 2620 400}%
\special{pa 2620 2360}%
\special{dt 0.045}%
%
\special{pn 13}%
\special{pa 1780 2360}%
\special{pa 1780 400}%
\special{dt 0.045}%
\end{picture}%
\caption{The toroidal grid graph $Q(6,9)$ covered by $Q(3,3)$'s with even kernel graphs}
\label{vicscover}
\end{figure}

Note that $Q(m,n)$ can be ``covered" by $Q(c,c)$'s,
and hence, we can obtain an even kernel graph of $Q(m,n)$ by combining that of $Q(c,c)$, 
as shown in Figure~\ref{vicscover}.
(Figure~\ref{vicscover} represents $Q(6,9)$ covered by six $Q(3,3)$'s 
with an even kernel graph shown in Figure~\ref{vics}.)
Therefore, since $Q(m,n)$ has an even kernel graph if $\gcd(m,n)=c > 1$,
the theorem holds by Lemma~\ref{lem_vic}.
\end{proof}

\begin{thm}\label{thm_torogrid_al2}
If $\gcd(2,n)=1$, then Alice can win the game on $Q(2,n)$.
\end{thm}
\begin{proof}
Without loss of generality, we set $(u_0,v_0)$ be a starting vertex.
Since $\gcd(2,n)=1$, $n$ is odd.
Alice first moves the token to $(u_0,v_1)$. 
After that, Alice plays the game according to Bob's move as follows:
\begin{itemize}
\item[(i)] If Bob moves the token to $(u_1,v_i)$ through an edge $(u_0,v_i)(u_1,v_i)$,
Alice moves it to $(u_0,v_i)$ using $(u_1,v_i)(u_0,v_i)$.
\item[(ii)] If Bob moves the token to $(u_0,v_{i+1})$,
then Alice moves it to $(u_0,v_{i+2})$,
where subscripts modulo $n$.
\end{itemize}
Observe that 
the strategy~(i) can be always applied
and that after the strategy (i) is applied,
Bob must move the token from $(u_0,v_{i})$ to $(u_0,v_{i+1})$.
Note that the index $i+1$ is always even when Alice uses the strategy~(ii).
Therefore, since $n$ is odd,
Alice finally moves the token from $(u_0,v_{n-1})$ to $(u_0,v_0)$,
that is, she wins the game.
\end{proof}

\begin{thm}\label{thm_torogrid_al3}
If $\gcd(3,n)=1$, then Alice can win the game on $Q(3,n)$.
\end{thm}
\begin{proof}
Without loss of generality, a starting vertex $s$ is $(u_0,v_0)$.
Moreover, by Theorem~\ref{thm_torogrid_al2}, we may assume that $n\geq 4$.

Alice first moves the token to $(u_1,v_0)$.
If Bob moves it to $(u_2,v_0)$, then Alice wins the game.
Thus Bob moves the token to $(u_1,v_1)$ by symmetry
and then Alice moves it to $(u_2,v_1)$.
Next Bob has to move the token to $(u_2,v_2)$
(otherwise Alice can move it back to $s$)
and then Alice moves it to $(u_0,v_2)$.
After that, Alice plays the game according to Bob's move 
until the token is moved to $(u_i,v_{n-2})$ for some $i \in \{0,1,2\}$ by herself, 
as follows
(where $i$ and $j$ in the following are modulo $3$ and $n$, respectively):
\begin{itemize}
\item[(i)] If Bob moves the token on $(u_i,v_j)$ to $(u_i,v_{j-1})$,
then Alice moves it to $(u_{i},v_{j-2})$.
\item[(ii)] If Bob moves the token on $(u_i,v_j)$ to $(u_{i+1},v_{j})$,
then Alice moves it to $(u_{i+1},v_{j-1})$.
\item[(iii)] If Bob moves the token on $(u_i,v_j)$ to $(u_i,v_{j+1})$
then Alice moves it to $(u_{i+1},v_{j+1})$.
\end{itemize}

Observe that 
in the above beginning moves from the starting vertex to $(u_0,v_2)$,
Alice applies only the strategy~(iii) twice except her first move.

In the strategy~(i), after Alice's move,
$(u_{i},v_{j-2})$ is incident to the unique edge $(u_{i-1},v_{j-2})(u_{i},v_{j-2})$
unless $(u_{i},v_{j-2}) = (u_0,v_0)$,
since two edges $(u_i,v_{j-3})(u_i,v_{j-2})$ and $(u_i,v_{j-2})(u_{i+1},v_{j-2})$ 
are used by the moves in~(iii).
Similarly, for the strategy~(ii),
$(u_{i+1},v_{j-1})$ is incident to the unique edge $(u_{i},v_{j-1})(u_{i+1},v_{j-1})$.
Thus, after Alice's move by the strategy~(i) (resp.,~(ii)),
Bob must move the token to $(u_{i-1},v_{j-2})$ (resp., $(u_{i},v_{j-1})$).
Hereafter,
Alice moves the token to $(u_{i-1},v_{j-3})$ (resp., $(u_{i},v_{j-2})$)
and then the same situation occurs of the current vertex.
Hence by applying the above move repeatedly,
the token is finally carried to $s$ from $(u_0,v_1)$ by Alice.

Therefore, we may suppose that 
until Alice moves the token to $(u_i,v_{n-2})$ for some $i \in \{0,1,2\}$ by herself, 
she always applies the strategy~(iii), that is,
two indices $i$ and $j$ of a current vertex $(u_i,v_j)$
are alternately increased one by one by Alice and Bob, respectively.
Therefore, we may assume that Alice finally moves the token 
to $(u_0,v_{n-2})$ (resp., $(u_1,v_{n-2})$) from $(u_2,v_{n-2})$ (resp., $(u_0,v_{n-2})$)
depending on $n$;
otherwise, 
i.e., if Alice finally moves the token $(u_2,v_{n-2})$ from $(u_1,v_{n-2})$,
then $n-2 \equiv 1 \pmod{3}$ and hence $n \equiv 0 \pmod{3}$, 
contrary to $\gcd(3,n) = 1$.

Thus the token is now put on $(u_0,v_{n-2})$ or $(u_1,v_{n-2})$.
In the former case, Bob moves to $(u_0,v_{n-1})$ and then Alice wins the game 
by moving it back to $s$.
In the latter case, 
Bob moves to $(u_1,v_{n-1})$ and then Alice moves it to $(u_1,v_0)$.
After that, since Bob must move the token to $(u_2,v_0)$,
Alice wins the game by moving it from $(u_2,v_0)$ to $s$.
Therefore, the theorem holds.
\end{proof}

\begin{thm}\label{thm_torogrid_no_victory}
If $\gcd(m,n)=1$, then there exists no even kernel graph of $Q(m,n)$.
\end{thm}
\begin{proof}
Let $\text{Ev}(m,n) \subseteq Q(m,n)$ be an even kernel graph of $Q(m,n)$.
From the definition, any vertex in the white part of $\text{Ev}(m,n)$, 
denoted by $W(m,n)$, 
has two or four neighbours\footnote{A vertex in $W(m,n)$ can have no neighbour, 
but in this case we can remove it from $\text{Ev}(m,n)$.}
and they are in the black part of $\text{Ev}(m,n)$, denoted by $B(m,n)$.
A {\em stopgap} of $\text{Ev}(m,n)$ is 
a vertex in $W(m,n)$ of degree~2 such that its neighbours lie on the same row or column.
When we ignore all stopgaps, 
$\text{Ev}(m,n)$ has several components surrounded by vertices in $W(m,n)$.
Note that any vertex in $B(m,n)$ cannot be adjacent to vertices not in $W(m,n)$.
We denote a component and stopgaps which are its neighbours (if exist) 
as a \textit{cluster} (see Figure~\ref{fig:clusters}).
In Figure~\ref{fig:clusters}, black vertices are in $B(m,n)$, 
gray vertices with bold circle are in $W(m,n)$, 
and gray vertices without edges are not in $\text{Ev}(m,n)$.

\begin{figure}[htb]
\begin{center}
\includegraphics[scale=0.6]{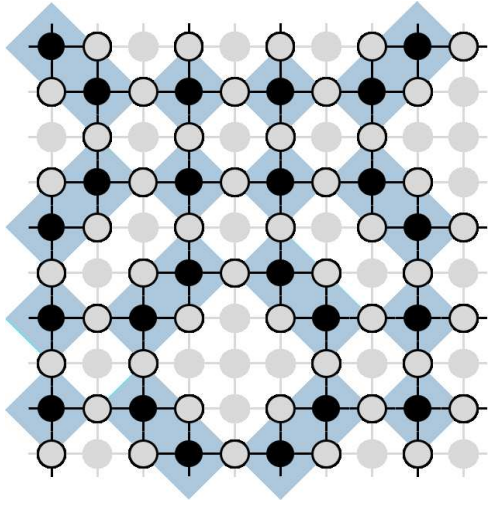}
\caption{An even kernel graph $\text{Ev}(10,10)$ of $Q(10,10)$ 
and its clusters denoted by shaded regions}
\label{fig:clusters}
\end{center}
\end{figure}

Every cluster looks a rectangle rotated 45 degrees.
This means that a cluster has four sides consisting of diagonally consecutive vertices in $W(m,n)$. 
For clusters, we have following claims.

\begin{cl}\label{prop_cluster_rectangle}
Every clusters are rectangles unless ${\rm Ev}(m,n)=Q(m,n)$.
\end{cl}
\begin{proof}
Assume that a cluster $C$ is not a rectangle. 
Then there must exist a vertex in $W(m,n) \subset C$ which is not a stopgap, 
and is adjacent to a vertex not in $\text{Ev}(m,n)$ 
and odd number of vertices in $B(m,n)$
(since all vertices in $B(m,n)$ are of degree~4).
This contradicts the definition of $\text{Ev}(m,n)$.
\end{proof}


\begin{cl}\label{fact_cluster_dual}
Any component in $Q(m,n) \setminus {\rm Ev}(m,n)$ 
induces a rectangular region. 
\end{cl}
\begin{proof}
This follows from the similar discussion 
in the proof of Claim~\ref{prop_cluster_rectangle}.
(For an example of such regions,
see a white region shown in Figure~\ref{fig:clusters}.)
\end{proof}

Using these claims, 
$\text{Ev}(m,n)$ must be $Q(m,n)$ when $\text{gcd}(m,n) = 1$ if exists.
On the other hand, $Q(m,n)$ is not bipartite when $\text{gcd}(m,n) = 1$.
This contradicts the definition of $\text{Ev}(m,n)$.
Therefore there exists no even kernel graph of $Q(m,n)$ if $\text{gcd}(m,n) = 1$
\end{proof}

Under results obtained above,
we conclude the paper with proposing the following conjecture
which implies that
Alice can win the feedback game on $Q(m,n)$ if and only if $\gcd(m,n)=1$.

\begin{conj}\label{conj_torogrid}
Alice can win the game on $Q(m,n)$ if $\gcd(m,n)=1$.
\end{conj}

\subsection*{Acknowledgement}
We express our appreciation to Professor Shozo Okada with taking this opportunity,
since the results in Section~\ref{sec_trigrid} are based on
his document for a special class.
We also appreciate Kazuki Kurimoto who is a graduate student in Kyoto Sangyo University,
giving us a nice proof of the feedback game on the triangular grid $T_5$.

\section*{Appendix}

Since it is clear that 
Alice wins the feedback game on $T_1$,
we shall prove that Alice wins the game on $T_5$
with starting vertex $s = v^0_0$.

Without loss of generality, Alice first moves the token to $v^1_0$,
and then Bob moves it to either (i) $v^2_0$ or (ii) $v^2_1$.
In the case (i) (resp., (ii)), 
Alice next moves the token to $v^3_0$ (resp., $v^2_2$).
For the case (i),
as shown the left of Figure~\ref{TGG5ev},
we can construct a ``good" bipartite subgraph for Alice;
note that Alice can move the token to a black vertex in the remaining game
as in the argument of the even kernel.
Therefore, Alice can finally move the token back to the starting vertex $s$.

\begin{figure}[htb]
\centering
\unitlength 0.1in
\begin{picture}( 47.0000, 22.3000)( 13.5000,-32.5000)
%
\special{pn 8}%
\special{sh 1.000}%
\special{ar 2400 1196 50 50  0.0000000 6.2831853}%
%
\special{pn 8}%
\special{sh 1.000}%
\special{ar 2000 1996 50 50  0.0000000 6.2831853}%
%
\special{pn 8}%
\special{sh 1.000}%
\special{ar 1800 2396 50 50  0.0000000 6.2831853}%
%
\special{pn 8}%
\special{sh 1.000}%
\special{ar 1400 3196 50 50  0.0000000 6.2831853}%
%
\special{pn 8}%
\special{sh 1.000}%
\special{ar 2200 3196 50 50  0.0000000 6.2831853}%
%
\special{pn 8}%
\special{sh 1.000}%
\special{ar 3400 3196 50 50  0.0000000 6.2831853}%
%
\special{pn 8}%
\special{sh 1.000}%
\special{ar 2800 1996 50 50  0.0000000 6.2831853}%
%
\special{pn 8}%
\special{sh 1.000}%
\special{ar 2800 2796 50 50  0.0000000 6.2831853}%
\put(25.4000,-11.0500){\makebox(0,0){$s$}}%
%
\special{pn 8}%
\special{sh 1.000}%
\special{ar 5000 1196 50 50  0.0000000 6.2831853}%
%
\special{pn 8}%
\special{sh 1.000}%
\special{ar 4800 3196 50 50  0.0000000 6.2831853}%
%
\special{pn 8}%
\special{sh 1.000}%
\special{ar 5000 1996 50 50  0.0000000 6.2831853}%
%
\special{pn 8}%
\special{sh 1.000}%
\special{ar 5400 2796 50 50  0.0000000 6.2831853}%
\put(51.4000,-11.0500){\makebox(0,0){$s$}}%
%
\special{pn 8}%
\special{pa 1800 2400}%
\special{pa 1400 3200}%
\special{fp}%
%
\special{pn 8}%
\special{pa 1400 3200}%
\special{pa 2200 3200}%
\special{fp}%
%
\special{pn 8}%
\special{pa 2200 3200}%
\special{pa 1800 2400}%
\special{fp}%
%
\special{pn 8}%
\special{pa 2200 3200}%
\special{pa 2600 3200}%
\special{fp}%
%
\special{pn 8}%
\special{pa 2600 3200}%
\special{pa 2800 2800}%
\special{fp}%
%
\special{pn 8}%
\special{pa 2800 2800}%
\special{pa 2400 2800}%
\special{fp}%
%
\special{pn 8}%
\special{pa 2400 2800}%
\special{pa 2200 3200}%
\special{fp}%
%
\special{pn 8}%
\special{pa 3000 3200}%
\special{pa 2800 2800}%
\special{fp}%
%
\special{pn 8}%
\special{pa 2800 2800}%
\special{pa 3200 2800}%
\special{fp}%
%
\special{pn 8}%
\special{pa 3200 2800}%
\special{pa 3400 3200}%
\special{fp}%
%
\special{pn 8}%
\special{pa 3400 3200}%
\special{pa 3000 3200}%
\special{fp}%
%
\special{pn 8}%
\special{pa 2800 2800}%
\special{pa 2600 2400}%
\special{fp}%
%
\special{pn 8}%
\special{pa 2600 2400}%
\special{pa 2800 2000}%
\special{fp}%
%
\special{pn 8}%
\special{pa 2800 2000}%
\special{pa 3000 2400}%
\special{fp}%
%
\special{pn 8}%
\special{pa 3000 2400}%
\special{pa 2800 2800}%
\special{fp}%
%
\special{pn 8}%
\special{pa 2800 2000}%
\special{pa 2000 2000}%
\special{fp}%
%
\special{pn 8}%
\special{pa 2800 2000}%
\special{pa 2400 1200}%
\special{fp}%
%
\special{pn 8}%
\special{pa 1800 2400}%
\special{pa 2200 2400}%
\special{fp}%
%
\special{pn 8}%
\special{pa 2200 2400}%
\special{pa 2000 2000}%
\special{fp}%
%
\special{pn 8}%
\special{pa 2400 1200}%
\special{pa 2200 1600}%
\special{dt 0.045}%
\special{sh 1}%
\special{pa 2200 1600}%
\special{pa 2248 1550}%
\special{pa 2224 1552}%
\special{pa 2212 1532}%
\special{pa 2200 1600}%
\special{fp}%
%
\special{pn 8}%
\special{pa 2200 1600}%
\special{pa 2020 1950}%
\special{dt 0.045}%
\special{sh 1}%
\special{pa 2020 1950}%
\special{pa 2068 1900}%
\special{pa 2044 1904}%
\special{pa 2034 1882}%
\special{pa 2020 1950}%
\special{fp}%
%
\special{pn 8}%
\special{pa 1990 2010}%
\special{pa 1820 2350}%
\special{dt 0.045}%
\special{sh 1}%
\special{pa 1820 2350}%
\special{pa 1868 2300}%
\special{pa 1844 2302}%
\special{pa 1832 2282}%
\special{pa 1820 2350}%
\special{fp}%
%
\special{pn 8}%
\special{sh 0}%
\special{ar 2200 2400 50 50  0.0000000 6.2831853}%
%
\special{pn 8}%
\special{sh 0}%
\special{ar 2600 2400 50 50  0.0000000 6.2831853}%
%
\special{pn 8}%
\special{sh 0}%
\special{ar 3000 2400 50 50  0.0000000 6.2831853}%
%
\special{pn 8}%
\special{sh 0}%
\special{ar 3200 2800 50 50  0.0000000 6.2831853}%
%
\special{pn 8}%
\special{sh 0}%
\special{ar 3000 3200 50 50  0.0000000 6.2831853}%
%
\special{pn 8}%
\special{sh 0}%
\special{ar 2600 3200 50 50  0.0000000 6.2831853}%
%
\special{pn 8}%
\special{sh 0}%
\special{ar 2400 2800 50 50  0.0000000 6.2831853}%
%
\special{pn 8}%
\special{sh 0}%
\special{ar 2000 2800 50 50  0.0000000 6.2831853}%
%
\special{pn 8}%
\special{sh 0}%
\special{ar 1600 2800 50 50  0.0000000 6.2831853}%
%
\special{pn 8}%
\special{sh 0}%
\special{ar 1800 3200 50 50  0.0000000 6.2831853}%
\put(21.4000,-14.0000){\makebox(0,0){$A$}}%
%
\special{pn 8}%
\special{sh 0}%
\special{ar 2600 1600 50 50  0.0000000 6.2831853}%
\put(19.9000,-17.2000){\makebox(0,0){$B$}}%
\put(17.8000,-21.5000){\makebox(0,0){$A$}}%
%
\special{pn 8}%
\special{pa 5000 1200}%
\special{pa 5200 1600}%
\special{fp}%
%
\special{pn 8}%
\special{pa 5200 1600}%
\special{pa 5000 2000}%
\special{fp}%
%
\special{pn 8}%
\special{pa 5000 2000}%
\special{pa 4800 2400}%
\special{fp}%
\special{pa 4800 2400}%
\special{pa 4800 2400}%
\special{fp}%
%
\special{pn 8}%
\special{pa 4800 2400}%
\special{pa 4400 2400}%
\special{fp}%
%
\special{pn 8}%
\special{pa 4400 2400}%
\special{pa 4600 2000}%
\special{fp}%
%
\special{pn 8}%
\special{pa 4600 2000}%
\special{pa 5000 2000}%
\special{fp}%
%
\special{pn 8}%
\special{pa 4400 2400}%
\special{pa 4000 3200}%
\special{fp}%
%
\special{pn 8}%
\special{pa 4000 3200}%
\special{pa 4800 3200}%
\special{fp}%
%
\special{pn 8}%
\special{pa 4800 3200}%
\special{pa 4400 2400}%
\special{fp}%
%
\special{pn 8}%
\special{pa 4800 3200}%
\special{pa 5200 3200}%
\special{fp}%
%
\special{pn 8}%
\special{pa 5200 3200}%
\special{pa 5400 2800}%
\special{fp}%
%
\special{pn 8}%
\special{pa 5400 2800}%
\special{pa 5000 2800}%
\special{fp}%
%
\special{pn 8}%
\special{pa 5000 2800}%
\special{pa 4800 3200}%
\special{fp}%
%
\special{pn 8}%
\special{pa 5400 2800}%
\special{pa 5600 3200}%
\special{fp}%
%
\special{pn 8}%
\special{pa 5600 3200}%
\special{pa 6000 3200}%
\special{fp}%
%
\special{pn 8}%
\special{pa 6000 3200}%
\special{pa 5800 2800}%
\special{fp}%
%
\special{pn 8}%
\special{pa 5800 2800}%
\special{pa 5400 2800}%
\special{fp}%
%
\special{pn 8}%
\special{pa 5400 2800}%
\special{pa 5200 2400}%
\special{fp}%
%
\special{pn 8}%
\special{sh 1.000}%
\special{ar 4000 3200 50 50  0.0000000 6.2831853}%
%
\special{pn 8}%
\special{sh 1.000}%
\special{ar 4400 2400 50 50  0.0000000 6.2831853}%
%
\special{pn 8}%
\special{sh 1.000}%
\special{ar 6000 3200 50 50  0.0000000 6.2831853}%
%
\special{pn 8}%
\special{sh 0}%
\special{ar 4800 2400 50 50  0.0000000 6.2831853}%
%
\special{pn 8}%
\special{sh 0}%
\special{ar 4600 2000 50 50  0.0000000 6.2831853}%
%
\special{pn 8}%
\special{sh 0}%
\special{ar 4600 2800 50 50  0.0000000 6.2831853}%
%
\special{pn 8}%
\special{sh 0}%
\special{ar 4200 2800 50 50  0.0000000 6.2831853}%
%
\special{pn 8}%
\special{sh 0}%
\special{ar 5000 2800 50 50  0.0000000 6.2831853}%
%
\special{pn 8}%
\special{sh 0}%
\special{ar 5800 2800 50 50  0.0000000 6.2831853}%
%
\special{pn 8}%
\special{sh 0}%
\special{ar 5600 3200 50 50  0.0000000 6.2831853}%
%
\special{pn 8}%
\special{sh 0}%
\special{ar 5200 3200 50 50  0.0000000 6.2831853}%
%
\special{pn 8}%
\special{sh 0}%
\special{ar 5200 1600 50 50  0.0000000 6.2831853}%
%
\special{pn 8}%
\special{pa 5006 1200}%
\special{pa 4806 1600}%
\special{dt 0.045}%
\special{sh 1}%
\special{pa 4806 1600}%
\special{pa 4854 1550}%
\special{pa 4830 1552}%
\special{pa 4818 1532}%
\special{pa 4806 1600}%
\special{fp}%
%
\special{pn 8}%
\special{pa 4800 1600}%
\special{pa 4980 1950}%
\special{dt 0.045}%
\special{sh 1}%
\special{pa 4980 1950}%
\special{pa 4968 1882}%
\special{pa 4956 1904}%
\special{pa 4932 1900}%
\special{pa 4980 1950}%
\special{fp}%
\put(47.4500,-14.0000){\makebox(0,0){$A$}}%
\put(48.0000,-18.3000){\makebox(0,0){$B$}}%
%
\special{pn 8}%
\special{pa 5020 1996}%
\special{pa 5370 1996}%
\special{dt 0.045}%
\special{sh 1}%
\special{pa 5370 1996}%
\special{pa 5304 1976}%
\special{pa 5318 1996}%
\special{pa 5304 2016}%
\special{pa 5370 1996}%
\special{fp}%
\put(52.2000,-19.0500){\makebox(0,0){$A$}}%
\put(15.4000,-13.1000){\makebox(0,0){(i)}}%
\put(41.4000,-13.1000){\makebox(0,0){(ii)-(a)}}%
%
\special{pn 8}%
\special{pa 5200 2400}%
\special{pa 5000 2000}%
\special{fp}%
%
\special{pn 8}%
\special{sh 0}%
\special{ar 5200 2400 50 50  0.0000000 6.2831853}%
%
\special{pn 8}%
\special{sh 0}%
\special{ar 2400 2000 50 50  0.0000000 6.2831853}%
%
\special{pn 8}%
\special{pa 5400 2000}%
\special{pa 5600 2400}%
\special{dt 0.045}%
\special{sh 1}%
\special{pa 5600 2400}%
\special{pa 5588 2332}%
\special{pa 5576 2352}%
\special{pa 5552 2350}%
\special{pa 5600 2400}%
\special{fp}%
%
\special{pn 8}%
\special{pa 5600 2410}%
\special{pa 5430 2760}%
\special{dt 0.045}%
\special{sh 1}%
\special{pa 5430 2760}%
\special{pa 5478 2710}%
\special{pa 5454 2712}%
\special{pa 5442 2692}%
\special{pa 5430 2760}%
\special{fp}%
\put(56.4000,-26.0000){\makebox(0,0){$A$}}%
\put(56.1000,-21.7000){\makebox(0,0){$B$}}%
\end{picture}%

\caption{Good bipartite subgraphs for the cases (i) and (ii)-(a)}
\label{TGG5ev}
\end{figure}
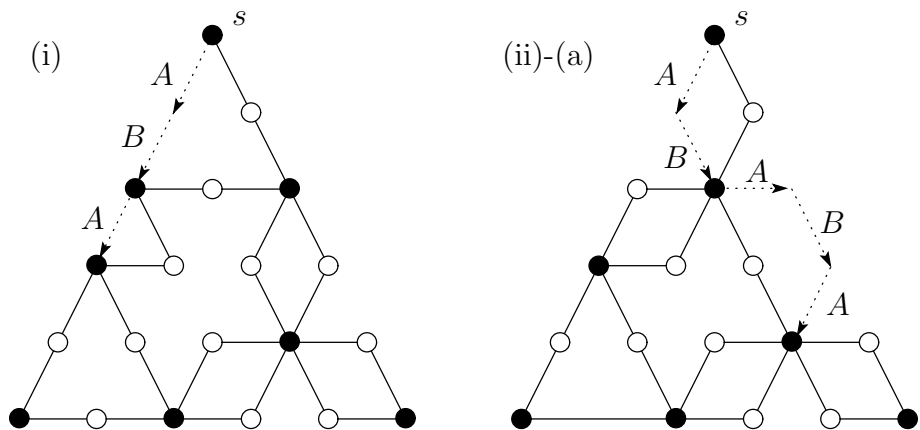

We divide the case (ii) to two subcases;
(a) Bob moves the token to $v^3_3$, or
(b) he moves the token to $v^3_2$.
In the former case, 
as shown in the right of Figure~\ref{TGG5ev},
Alice wins the game similarly to the case (i).
In the latter case, Alice moves the token to $v^4_3$.
If Bob moves the token to $v^4_4$ or $v^5_4$,
then Alice can move it back to $v^4_3$ along a 4-cycle $v^4_3v^4_4v^5_5v^5_4$.
Moreover, if Bob moves the token to $v^3_3$,
then Alice can win the game by moving it to $v^2_2$ 
(since Bob must move it to $v^1_1$ in his next move).
Such a vertex $u$ (to which a player loses the game by moving the token) is called a {\em dead vertex}
(see Figure~\ref{TGG5ev2}; a dead vertex is marked by `d' and colored by gray).
Thus Bob moves the token from $v^4_3$ to either (1) $v^5_3$ or (2) $v^4_2$.

\begin{figure}[htb]
\centering
\unitlength 0.1in
\begin{picture}( 21.0000, 22.3000)( 21.5000,-32.5000)
%
\special{pn 8}%
\special{sh 1.000}%
\special{ar 3200 1196 50 50  0.0000000 6.2831853}%
%
\special{pn 8}%
\special{sh 1.000}%
\special{ar 3000 3196 50 50  0.0000000 6.2831853}%
%
\special{pn 8}%
\special{sh 1.000}%
\special{ar 3200 1996 50 50  0.0000000 6.2831853}%
%
\special{pn 8}%
\special{sh 1.000}%
\special{ar 3600 2800 50 50  0.0000000 6.2831853}%
\put(33.4000,-11.0500){\makebox(0,0){$s$}}%
%
\special{pn 8}%
\special{pa 3600 2800}%
\special{pa 3800 3200}%
\special{fp}%
%
\special{pn 8}%
\special{pa 3800 3200}%
\special{pa 4200 3200}%
\special{fp}%
%
\special{pn 8}%
\special{pa 4200 3200}%
\special{pa 4000 2800}%
\special{fp}%
%
\special{pn 8}%
\special{pa 4000 2800}%
\special{pa 3600 2800}%
\special{fp}%
%
\special{pn 8}%
\special{sh 1.000}%
\special{ar 2200 3200 50 50  0.0000000 6.2831853}%
%
\special{pn 8}%
\special{sh 1.000}%
\special{ar 2600 2400 50 50  0.0000000 6.2831853}%
%
\special{pn 8}%
\special{sh 1.000}%
\special{ar 4200 3200 50 50  0.0000000 6.2831853}%
%
\special{pn 8}%
\special{sh 1.000}%
\special{ar 3000 2400 50 50  0.0000000 6.2831853}%
%
\special{pn 8}%
\special{sh 1.000}%
\special{ar 2800 2800 50 50  0.0000000 6.2831853}%
%
\special{pn 8}%
\special{sh 1.000}%
\special{ar 2400 2800 50 50  0.0000000 6.2831853}%
%
\special{pn 8}%
\special{sh 0}%
\special{ar 4000 2800 50 50  0.0000000 6.2831853}%
%
\special{pn 8}%
\special{sh 0}%
\special{ar 3800 3200 50 50  0.0000000 6.2831853}%
%
\special{pn 8}%
\special{pa 3206 1200}%
\special{pa 3040 1540}%
\special{dt 0.045}%
\special{sh 1}%
\special{pa 3040 1540}%
\special{pa 3088 1490}%
\special{pa 3064 1492}%
\special{pa 3052 1472}%
\special{pa 3040 1540}%
\special{fp}%
%
\special{pn 8}%
\special{pa 3000 1600}%
\special{pa 3180 1950}%
\special{dt 0.045}%
\special{sh 1}%
\special{pa 3180 1950}%
\special{pa 3168 1882}%
\special{pa 3156 1904}%
\special{pa 3132 1900}%
\special{pa 3180 1950}%
\special{fp}%
\put(29.4500,-14.0000){\makebox(0,0){$A$}}%
\put(30.0000,-18.3000){\makebox(0,0){$B$}}%
%
\special{pn 8}%
\special{pa 3200 2000}%
\special{pa 3550 2000}%
\special{dt 0.045}%
\special{sh 1}%
\special{pa 3550 2000}%
\special{pa 3484 1980}%
\special{pa 3498 2000}%
\special{pa 3484 2020}%
\special{pa 3550 2000}%
\special{fp}%
\put(34.2000,-19.0500){\makebox(0,0){$A$}}%
%
\special{pn 8}%
\special{pa 3600 2000}%
\special{pa 3400 2400}%
\special{dt 0.045}%
\special{sh 1}%
\special{pa 3400 2400}%
\special{pa 3448 2350}%
\special{pa 3424 2352}%
\special{pa 3412 2332}%
\special{pa 3400 2400}%
\special{fp}%
%
\special{pn 8}%
\special{pa 3410 2410}%
\special{pa 3580 2760}%
\special{dt 0.045}%
\special{sh 1}%
\special{pa 3580 2760}%
\special{pa 3570 2692}%
\special{pa 3558 2712}%
\special{pa 3534 2710}%
\special{pa 3580 2760}%
\special{fp}%
\put(34.1000,-21.7000){\makebox(0,0){$B$}}%
%
\special{pn 8}%
\special{sh 1.000}%
\special{ar 3600 2000 50 50  0.0000000 6.2831853}%
%
\special{pn 8}%
\special{sh 1.000}%
\special{ar 3000 1600 50 50  0.0000000 6.2831853}%
%
\special{pn 8}%
\special{sh 1.000}%
\special{ar 2600 3200 50 50  0.0000000 6.2831853}%
%
\special{pn 8}%
\special{pa 3600 2800}%
\special{pa 3270 2800}%
\special{fp}%
\special{sh 1}%
\special{pa 3270 2800}%
\special{pa 3338 2820}%
\special{pa 3324 2800}%
\special{pa 3338 2780}%
\special{pa 3270 2800}%
\special{fp}%
%
\special{pn 8}%
\special{sh 0.400}%
\special{ar 3800 2400 50 50  0.0000000 6.2831853}%
%
\special{pn 8}%
\special{sh 0.400}%
\special{ar 3400 1600 50 50  0.0000000 6.2831853}%
%
\special{pn 8}%
\special{sh 0.400}%
\special{ar 2800 2000 50 50  0.0000000 6.2831853}%
\put(36.0000,-25.6000){\makebox(0,0){$A$}}%
\put(39.0000,-23.1000){\makebox(0,0){d}}%
\put(35.1000,-15.2000){\makebox(0,0){d}}%
%
\special{pn 8}%
\special{sh 1.000}%
\special{ar 3400 3200 50 50  0.0000000 6.2831853}%
%
\special{pn 8}%
\special{sh 1.000}%
\special{ar 3200 2800 50 50  0.0000000 6.2831853}%
\put(26.8000,-19.3000){\makebox(0,0){d}}%
%
\special{pn 8}%
\special{pa 3590 2830}%
\special{pa 3440 3130}%
\special{fp}%
\special{sh 1}%
\special{pa 3440 3130}%
\special{pa 3488 3080}%
\special{pa 3464 3082}%
\special{pa 3452 3062}%
\special{pa 3440 3130}%
\special{fp}%
\put(34.1000,-27.0000){\makebox(0,0){(2)}}%
\put(33.7000,-29.8000){\makebox(0,0){(1)}}%
\end{picture}%

\caption{The case (ii)-(b)}
\label{TGG5ev2}
\end{figure}
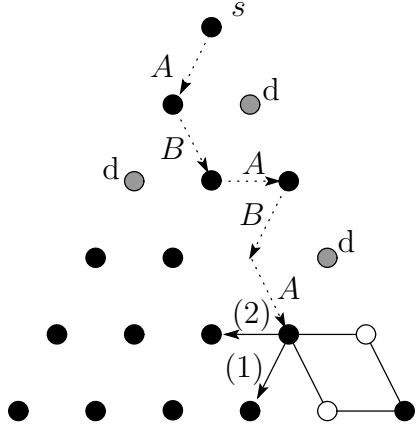

\medskip
\noindent
The proof of the case (1):
Alice moves the token to $v^5_2$.
By observing that $v^4_2$ and $v^5_1$ are dead,
Bob must move the token to $v^4_1$,
and then Alice moves it to $v^3_0$.
Since $v^4_0$ is also dead now, Bob moves the token to $v^3_1$.
Therefore, Alice can force Bob to move the token to a dead vertex, by moving it from $v^3_1$ to $v^4_1$.

\medskip
\noindent
The proof of the case (2):
Alice moves the token to $v^5_2$.
Similarly to the previous case, 
Bob must move the token to $v^4_1$ since $v^5_1$ and $v^5_3$ are dead.
After that, Alice can force Bob to move the token to a dead vertex
by using one of the following two patterns
($\to_A$ (resp., $\to_B$) means a move of the token by Alice (resp., Bob)):
\begin{enumerate}
\item $v^4_1 \to_A v^4_2 \to_B v^3_1 \to_A v^2_1 \to_B v^3_2 \to_A v^4_2$
\item $v^4_1 \to_A v^4_2 \to_B v^3_2 \to_A v^2_1 \to_B v^3_1 \to_A v^4_2$
\end{enumerate}

Therefore, Alice wins the feedback game on the triangular grid graph $T_5$.

\end{document}